\documentclass[11pt]{article}
\usepackage{palatino}
\usepackage{pgfplots}
\usepackage{tikz, tikzpeople}
\usetikzlibrary{arrows}
\usetikzlibrary{shapes.multipart}
\usetikzlibrary{fit}
\usetikzlibrary{shapes.geometric,positioning}
\usetikzlibrary{calc}
\usepackage{bm}
\usepackage[utf8]{inputenc}
\usepackage{mathtools}
\usepackage{graphicx}
\usepackage{algorithm}
\usepackage{cite}
\usepackage{amsmath,amssymb,amsfonts, amsthm,  mathrsfs}
\usepackage{algpseudocode}
\usepackage{textcomp}
\usepackage{xcolor}
\usepackage{enumitem}
\usepackage{cancel}
\usepackage{hyperref}
\usepackage{kbordermatrix}

\newcommand{\indep}{\perp\!\!\perp}

\usetikzlibrary{patterns}
\newtheorem{lemma}{Lemma}
\newtheorem{theorem}{Theorem}

\newtheorem{remark}{Remark}
\newtheorem{corollary}{Corollary}
\newtheorem{definition}{Definition}

\def\bbsmatrix#1{\begin{bsmallmatrix}#1\end{bsmallmatrix}}
\newcommand{\algrule}[1][.2pt]{\par\vskip.5\baselineskip\hrule height #1\par\vskip.5\baselineskip}
\newlength{\leftstackrelawd}
\newlength{\leftstackrelbwd}
\def\leftstackrel#1#2{\settowidth{\leftstackrelawd}%
{${{}^{#1}}$}\settowidth{\leftstackrelbwd}{$#2$}%
\addtolength{\leftstackrelawd}{-\leftstackrelbwd}%
\leavevmode\ifthenelse{\lengthtest{\leftstackrelawd>0pt}}%
{\kern-.5\leftstackrelawd}{}\mathrel{\mathop{#2}\limits^{#1}}}

\DeclareMathOperator*{\rank}{rank}
\DeclareMathOperator*{\minrk}{minrk}
\DeclareMathOperator*{\tri}{tri}
\DeclareMathOperator*{\diag}{diag}
\DeclareMathOperator*{\NS}{{NS}}
\DeclareMathOperator*{\semidet}{{semi-det}}

\allowdisplaybreaks

\title{Can Non-Signaling Assistance Increase the Degrees of Freedom of a Wireless Network?}
\author{Yuhang Yao, Syed A. Jafar\\
{\small Center for Pervasive Communications and Computing (CPCC)}\\
{\small University of California Irvine, Irvine, CA 92697}\\
{\small \it Email: \{yuhangy5, syed\}@uci.edu}
}

\date{}

\begin{document}
\maketitle

\begin{abstract}
An open question recently posed by Fawzi and Ferme [IEEE Transactions on Information Theory 2024], asks whether non-signaling (NS) assistance can increase the capacity of a broadcast channel (BC). We answer this question in the affirmative, by showing that for a certain $K$-receiver BC model, called Coordinated Multipoint broadcast (CoMP BC) that arises naturally in wireless networks, NS-assistance provides multiplicative gains in both capacity and degrees of freedom (DoF), even achieving $K$-fold improvements in extremal cases. Somewhat surprisingly, this is shown to be true even for $2$-receiver broadcast channels that are semi-deterministic and/or degraded. In a CoMP BC, $B$ single-antenna transmitters, supported by a backhaul that allows them to share data, act as one $B$-antenna transmitter, to send independent messages to $K$ receivers, each equipped with a single receive antenna. A fixed and globally known connectivity matrix  specifies  for each transmit antenna, the subset of receivers that are connected to (have a non-zero channel coefficient to) that antenna. 
Besides the connectivity, there is no channel state information at the transmitter. The receivers have perfect channel knowledge. We show that NS-assistance has no DoF advantage in a fully connected CoMP BC. The DoF region is fully characterized for a class of connectivity patterns associated with tree graphs, for which the classical sum-DoF value is shown to be the number of leaf nodes, while the NS-assisted sum-DoF value is the total number of all (non-root) nodes. 
For arbitrary connectivity patterns, the sum-capacity with NS-assistance is  bounded above and below by the min-rank and triangle number of the connectivity matrix, respectively, leading to matching bounds in many cases, e.g., if $\min(B,K)\leq 6$. 
While translations to Gaussian settings are demonstrated, for simplicity most of our results are presented under noise-free, finite-field $(\mathbb{F}_q)$ models. Converse proofs for classical DoF are found by adapting  the Aligned Images bounds to the finite field model. Converse bounds for NS-assisted DoF/capacity extend the same-marginals property to the BC with NS-assistance available to all parties.  Beyond the BC setting, even stronger (unbounded) gains in capacity due to NS-assistance are established  for certain `communication with side-information' settings, such as the fading dirty paper channel.
\end{abstract}
{\let\thefootnote\relax\footnote{Presented in part at the IEEE International Symposium on Information Theory (ISIT) 2025.}\addtocounter{footnote}{-1}}
\newpage
\section{Introduction}
Understanding the prospects of new  technologies (e.g., the quantum internet \cite{caleffi_tutorial2}) requires a re-evaluation of long-established capacity limits in information theory, especially when critical underlying assumptions have to be relaxed. Consider \emph{nonlocality} \cite{PRbox}, as represented by the idea of `\emph{non-signaling (NS) assistance}.' NS-assistance in a network of communication channels is essentially a catch-all framework that allows free access to any resource in addition to those channels, provided that the resource \emph{by itself} (without the use of the channels) does not allow any communication in the network \cite{fawzi2024MAC}. NS-assistance includes all shared multipartite quantum entanglements within its scope. NS-assistance also allows \emph{more} than what is possible with quantum physics, i.e., potential super-quantum theories that may emerge in the future, disallowing only that which is strictly forbidden by the theory of special relativity, namely that information cannot be transmitted  instantaneously (faster than the speed of light). The intriguing contrast between its inherently restrictive (a NS resource is useless for communication by itself) and inclusive (contains quantum entanglement as a special case) features prompts the question \cite{PRbox, Notzel}: \emph{how can NS-assistance improve the capacity of communication channels?}

\subsection{NS-Assisted Capacity Improvements in Prior Works}
If instead of \emph{communication}, the goal was distributed \emph{computation}, then it is known that NS-resources are much too strong, e.g., NS-assistance allows all distributed decision problems to be solved with only one bit of communication \cite{van2013implausible}.
However, for \emph{communication} tasks, the utility of NS resources is not well understood. On one hand, there are several capacity metrics, such as zero-error capacity, arbitrarily varying channel capacity, and maximum coding rate under finite  block-length and error probability constraints, by which NS-assistance (and even its quantum restriction) has been shown to be tremendously beneficial \cite{cubitt2011zero,cubitt2010improving,Notzel}. For example, Cubitt et al. \cite{cubitt2011zero,cubitt2010improving} show that the \emph{zero-error capacity} with non-signaling assistance, $C_0^{\NS}$, can be arbitrarily larger than the zero-error capacity without non-signaling assistance, $C_0$, even providing examples where the latter is zero while the former is non-zero. On the other hand,  in terms of the widely studied \emph{Shannon capacity} metric (requiring vanishing error guarantees for asymptotically large blocklengths)  \cite{shannon1948mathematical}, for a \emph{point to point discrete memoryless channel,} it is known \cite{Bennett_Shor_Smolin_Thapliyal_PRL,matthews2012linear,Barman_Fawzi} that NS-assistance (and therefore quantum entanglement) offers \emph{no advantage at all.}

Beyond the point to point setting,  for discrete memoryless communication \emph{networks}, it is known \cite{Quek_Shor, leditzky2020playing,  seshadri2023separation, pereg2024MAC_QEassist} that NS-assistance \emph{can} increase the Shannon\footnote{By Shannon capacity of a \emph{network}, we refer to its \emph{sum}-rate capacity.} capacity.
Quek and Shor in \cite{Quek_Shor} provide examples of interference channels, where the NS-assisted capacity region is strictly larger than the quantum-entanglement assisted capacity region, which in turn is strictly larger than the classical capacity region. Leditzky et al. \cite{leditzky2020playing} provide examples to show that transmitter-side quantum entanglement in a MAC can strictly enlarge the capacity region. The bounds from \cite{leditzky2020playing} are improved  by Seshadri et al. in \cite{seshadri2023separation} to show that the sum rate increases from at most $3.02$ to $3.17$ bits per transmission. 
 While characterizing NS-correlations in communication networks is recognized as a hard problem in general \cite{NS_Gisin}, a key idea for showing such improvements is to map the winning strategy of a multiplayer pseudo-telepathy game,\footnote{In a pseudo-telepathy game, spatially separated players, who cannot communicate once the game starts, each receive an input and must output an answer such that a joint winning condition is satisfied. It has been shown that there exist games for which no classical strategy can achieve perfect success, yet they can be won with certainty if the players have pre-shared quantum entanglement, conduct input-dependent measurement on their respective quantum systems, and provide answers based on the outputs of the measurements.} such as the magic square \cite{Magicsquare} game, into a coding strategy in a communication network.
 A downside is that the resulting channel models can be too artificial, making the new insights difficult to translate to communication networks  that are commonly encountered, e.g., wireless networks. The following questions  (the emphasis is ours) are posed by Fawzi and Ferme in \cite{fawzi2024MAC}.
\begin{enumerate}
\item[Q1.] ``\emph{Can NS correlations lead to \underline{significant} gains in capacity for \underline{natural} [communication networks\footnote{\label{note1}The questions in \cite{fawzi2024MAC} focus on \emph{MACs}, but are equally justified for \emph{any} `natural' communication network.}\addtocounter{footnote}{-1}\addtocounter{Hfootnote}{-1}]?}"
\item[Q2.] ``\emph{Can we find [\ldots] the \underline{capacity region}\footnotemark ~[\ldots] when NS resources between the parties are allowed?}" 
\end{enumerate}
For multiple access channels (MACs),  significant progress made in \cite{fawzi2024MAC,pereg2024MAC_QEassist} shows, e.g., that NS-assistance does increase the Shannon capacity of a binary adder MAC, which is a {\it natural} model reflecting the superposition property of a wireless uplink. Notably,  the Shannon capacity improvements noted thus far have been relatively modest, e.g.,  from 1.5 to 1.5425 bits/transmission in \cite{fawzi2024MAC}, and the \emph{capacity region} remains open even for the NS-assisted binary-adder MAC. Also, it is known that NS-assistance provides no advantage in Shannon capacity in a MAC \cite{fawzi2024MAC} if independent NS resources are shared pairwise between each transmitter and the receiver.  

In \cite{fawzi2024broadcast}, Fawzi and Ferme study NS-assisted Shannon capacity of broadcast channels. Noting that the BC setting is more challenging than the MAC, \cite{fawzi2024broadcast} establishes two negative (impossibility) results -- 1) that NS-assistance provides no capacity advantage in a BC if the NS resource is shared only among the receivers, and 2) that NS-assistance provides no capacity advantage in a deterministic BC. The study in  \cite{fawzi2024broadcast} concludes with open problems that include the following.
\begin{enumerate}
\item[Q3.] Can NS-assistance improve the capacity region of a \underline{semi-deterministic} and/or \underline{degraded} broadcast channel? Reference \cite{fawzi2024broadcast} hints that the answer is likely to be negative, ``\emph{\ldots could be a crucial first step toward showing that the capacity region for those classes [semi-deterministic or degraded BC]  is the \underline{same with or without NS assistance}.}" 
\item[Q4.] Can NS-assistance improve the capacity region of a \underline{general} broadcast channel? Here, \cite{fawzi2024broadcast} suggests that the answer may be affirmative, ``\emph{\ldots full non-signaling assistance between the three parties could improve the capacity region of general broadcast channels, which is left as a \underline{major} \underline{open} \underline{question}.}"
\end{enumerate}
In this work we answer the aforementioned questions in the affirmative. For a broad class of  BC settings (called Coordinated Multipoint (CoMP) \cite{CoMP}) that arise \emph{naturally} in wireless networks, we characterize the exact \emph{capacity region} with NS-assistance allowed among all parties, and  demonstrate \emph{significant} (multiplicative) gains in capacity (as well as  DoF\footnote{DoF stands for \emph{degrees of freedom}. Informally, the DoF value (formally defined in Sections \ref{sec:Fqmodel} and \ref{sec:gaussianmodel}) represents the ratio of a network’s sum-rate capacity to the capacity of a point to point channel, in the limit of large alphabet (large field-size $q$ for $\mathbb{F}_q$ models (Section \ref{sec:Fqmodel}), and large transmit power $P$ for Gaussian models (Section \ref{sec:gaussianmodel})). }), even when the setting corresponds to a \emph{semi-deterministic} BC. In fact, we show that NS-assistance improves the sum-capacity in the strongest way possible, in the following sense --- there exist both semi-deterministic BCs and degraded BCs (with $2$ receivers) where the capacity advantage due to NS-assistance is either equal to or arbitrarily close to a factor of $2$, and there exist general $K$-user BCs where the advantage is arbitrarily close to a factor of $K$.

\subsection{NS-Assistance in a Wireless Network}
\begin{figure}[t]
\center
\begin{tikzpicture}
\node (myfirstpic) at (0,0) {\includegraphics[width=0.6\textwidth]{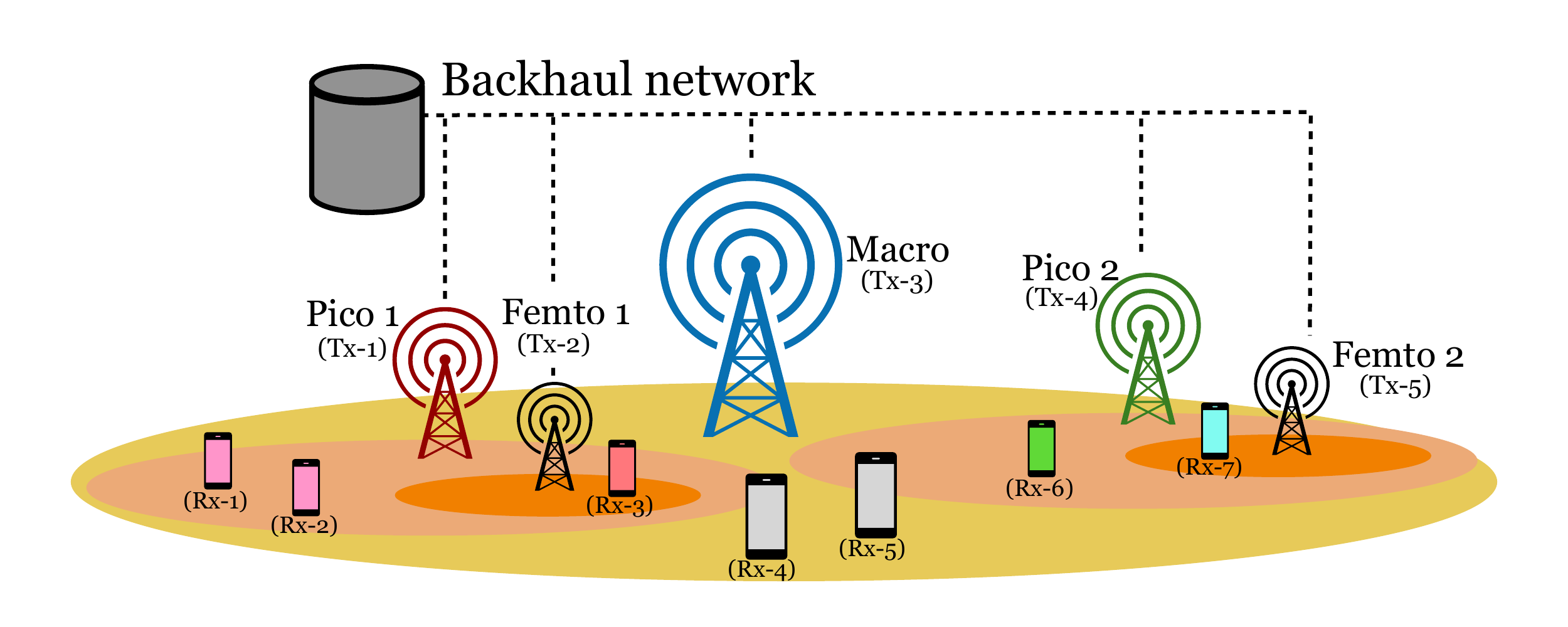}};
\node at (8.1,0){\(
   \kbordermatrix{
    & \mbox{\tiny Tx-$1$} & \mbox{\tiny Tx-$2$} & \mbox{\tiny Tx-$3$} & \mbox{\tiny Tx-$4$}& \mbox{\tiny Tx-$5$} \\
    \mbox{\tiny Rx-$1$} & * & 0 & * & 0 &0 \\
    \mbox{\tiny Rx-$2$} & * & 0 & * & 0 &0 \\
    \mbox{\tiny Rx-$3$} & * & * & * & 0 &0 \\
    \mbox{\tiny Rx-$4$} & 0 & 0 & * & 0  &0\\
    \mbox{\tiny Rx-$5$} & 0 & 0 & * & 0  &0\\
    \mbox{\tiny Rx-$6$} & 0 & 0 & * & *  &0\\
    \mbox{\tiny Rx-$7$} & 0 & 0 & * & *  &*\\
  }={\bf M}
  \)};
\end{tikzpicture}
\caption{A coordinated multipoint broadcast (CoMP BC) setting. A hierarchy of nested macro, pico and femto cells is apparent. Each transmit antenna (represented as a tower) can be heard everywhere within the boundaries of the  colored elliptical area (cell) in which it stands. Receivers are shown as smartphones.  Receivers in the figure are labeled from left to right as Rx-$1$, $\cdots$, Rx-$7$, while transmit antennas, also from left to right as shown, are labeled as Tx-$1$, $\cdots$, Tx-$5$. The $7\times 5$ connectivity matrix ${\bf M}$ has `$*$' entries for connected channels, and $0$'s otherwise.}
\label{fig:CoMP}
\end{figure}
A Coordinated Multipoint \cite{CoMP} Broadcast (CoMP BC) setting is illustrated in Figure \ref{fig:CoMP}. The network is composed of $B$ single-antenna base-station transmitters, Tx-$1$, $\cdots$, Tx-$B$, connected via a high-speed backhaul  network that allows them to share data, and thus act as a single $B$-antenna transmitter, in order to send independent messages to $K$ receivers (users), Rx-$1$, $\cdots$, Rx-$K$. A salient feature of such a network is its connectivity pattern. Because wireless connectivity is distance and transmit power dependent, in this CoMP BC, each transmit antenna can only be heard by receivers in its geographical vicinity, upto a range that depends on its transmit power. The connectivity pattern is fixed (globally  known), and specifies for each transmit antenna, the subset of receivers that are connected to (have a non-zero channel coefficient to) that antenna. In this work, we study such networks under the assumption that there is no channel state information at the transmitters (no CSIT) besides the fixed connectivity pattern.  Channel knowledge at receivers is assumed perfect (perfect CSIR). We explore the capacity and degrees of freedom (DoF) of such CoMP BC networks,  with and without NS-assistance, in order to gauge the advantages provided by the NS resources. As in the general BC framework, we assume that a non-signaling resource (such as shared randomness, entangled quantum systems, or any hypothetical  resource that creates post-quantum NS correlation, e.g., a PR box \cite{PRbox}) is pre-shared between the base-station transmitters and the receivers prior to the start of the communications.

Our results show (Theorem \ref{thm:fullconnect}) that NS-assistance has no DoF advantage  over classical schemes in a fully connected CoMP BC, where every transmit antenna can be heard by every receiver. However, for certain $K$-level hierarchical connectivity patterns (e.g., a femtocell within a picocell within a microcell within a macrocell is a $4$ level hierarchy (Figure \ref{fig:bestcaseexample})), the advantage of NS-assistance is found to be remarkably strong (Theorem \ref{thm:bestcase}), improving DoF by a multiplicative factor of $K$ relative to classical coding schemes. Surprisingly (in light of Q2), this is the case even for a semi-deterministic and/or degraded BC. Recall that Fawzi and Ferme have shown that NS-assistance provides no capacity gain in a deterministic BC \cite{fawzi2024broadcast}. We fully characterize (Theorem \ref{thm:tree_network}) the DoF region for \emph{tree networks} (Definition \ref{def:treenet}) --- a class of CoMP BC connectivity patterns associated with tree graphs (networks shown in Figure \ref{fig:CoMP} and \ref{fig:bestcaseexample} correspond to tree graphs), for which the classical sum-DoF value is shown (Corollary \ref{cor:sumDoF}) to be the number of leaf nodes, while the NS-assisted sum-DoF value is the total number of all (non-root) nodes, the extreme case being a path graph that has only one leaf node, representing the aforementioned vertical hierarchy. 
Sum-DoF (and sum-capacity) of CoMP BC with NS-assistance for arbitrary connectivity patterns are shown to be bounded above (Theorem \ref{thm:Sato_min_rank}) and below (Theorem \ref{thm:triangle_achi})   by the the min-rank (Definition \ref{def:minrk}) and triangle number  (Definition \ref{def:tnumber}) of the connectivity matrix ${\bf M}$, respectively, leading to matching bounds in many cases, e.g., if $\min(B,K)\leq 6$ (Corollary \ref{cor:upto6}). 
While translations to Gaussian settings are included (Theorem \ref{thm:tree_network}), for simplicity most of our results are presented via noise-free, finite-field $(\mathbb{F}_q)$ models. Converse bounds for classical DoF are found (Appendix \ref{sec:AIS}) by adapting  the Aligned Images argument \cite{davoodi2016aligned} to the finite field model. Converse bounds for NS-assisted DoF/capacity extend the same-marginals property to the BC with NS-assistance available to all parties (Theorem \ref{thm:same_marginal}). Beyond the CoMP BC, even stronger (unbounded) gains in capacity due to NS-assistance are shown to be possible for certain communication with side-information settings, such as the fading dirty paper channel (Theorem \ref{thm:fadingdirt}).

 \subsection{Notation}
$\mathbb{R}_+$ is the set of non-negative reals. $\mathbb{N}$ is the set of positive integers. For $n \in \mathbb{N}$, $[n] \triangleq \{1,2,\cdots, n\}$. $A^{[n]}$ is the compact notation for $[A^{(1)},A^{(2)},\cdots, A^{(n)}]$. $A_{[n]}$ is the compact notation for $[A_1,A_2,\cdots, A_n]$. The notation $\diag([a_1,a_2,\cdots, a_K])$ denotes the $K\times K$ diagonal matrix with elements $a_1,a_2,\cdots, a_K$ on the main diagonal.  $\mathbb{F}_q$ is the finite field with order $q$ being a power of a prime. $\mathbb{F}_q^\times$ is defined as $\mathbb{F}_q\setminus\{0\}$, i.e., the set of all non-zero elements of $\mathbb{F}_q$. We write $g(n) = o_n(f(n))$ if $\lim_{n\to \infty} \frac{g(n)}{f(n)}$ $  = 0$ and write $g(n) = O_n(f(n))$ if $\limsup_{n\to \infty} \frac{|g(n)|}{|f(n)|}  < \infty$. We write `Tx' for `Transmitter,' and `Rx-$k$' for `Receiver $k$.' We write Tx-$k$ for transmit antenna $k$. The terms `user' and `receiver' are used interchangeably. For random variables $A,B$, by $A\indep B$ we mean that $A$ is independent of $B$.

\section{Problem Formulation} \label{sec:BC}
We begin with a discrete memoryless broadcast setting, formalize the classical and NS-assisted coding frameworks, and then specialize to the wireless settings that are our main focus.
 \subsection{Discrete Memoryless $K$-user BC}\label{sec:dmcbc}
A discrete memoryless $K$-user broadcast channel is specified by a tuple $(\mathcal{X},(\mathcal{Y}_k)_{k\in[K]},\mathcal{N}_{Y_1\cdots Y_K\mid X})$ where $\mathcal{X}$ is the input alphabet, $\mathcal{Y}_k$ is the output alphabet for Rx-$k$, and the conditional probability of observing any $y_k\in\mathcal{Y}_k$ at each Rx-$k$, $k\in[K]$, for any input symbol $x\in\mathcal{X}$ sent from the Tx, is given by the channel distribution $\mathcal{N}_{Y_1\cdots Y_K \mid X}(y_1,\cdots, y_K\mid x)$. The \emph{marginal distribution} of the channel from the Tx to Rx-$k$, $k\in [K]$, is defined as,
\begin{align}
	\mathcal{N}_{Y_k\mid X}(y_k \mid x) \triangleq \sum_{l\in [K]\setminus\{k\}}\sum_{y_l \in \mathcal{Y}_l} \mathcal{N}_{Y_1\cdots Y_K\mid X}(y_1,\cdots, y_K\mid x).
\end{align}
There are $K$ independent messages, $W_1,\cdots, W_K$, which originate at the Tx, such that $W_k$ is the desired message for Rx-$k$, $k\in [K]$. Let $\mathcal{N}^{\otimes n}$, $n\in\mathbb{N}$, denote $n$   uses of the channel $\mathcal{N}$. Specifically, for the $\tau^{th}$ use, let the input to the channel be denoted as $X^{(\tau)}$ and outputs be denoted as $Y_k^{(\tau)}$ for $k\in [K]$. Given any input $X^{[n]} = x^{[n]}\in\mathcal{X}^n$, the discrete memoryless property of the channel defines the conditional distribution of the outputs $Y_k^{[n]} = y_k^{[n]}\in\mathcal{Y}_k^n$, $\forall k\in[K]$  as,
\begin{align} \label{eq:channel_extension}
	& \mathcal{N}^{\otimes n}(y_1^{[n]},\cdots, y_K^{[n]}\mid x^{[n]}) = \prod_{\tau=1}^n \mathcal{N}(y_1^{(\tau)},\cdots, y_K^{(\tau)}\mid x^{(\tau)}).
\end{align}

\subsection{Classical coding schemes}
A classical coding scheme operates over $n\in\mathbb{N}$ channel uses to transmit the $K$ messages, $W_1,W_2,\cdots,$ $ W_K$, that are distributed uniformly over the  non-empty sets $\mathcal{M}_1,\cdots, \mathcal{M}_K$, respectively. The scheme specifies an encoder $\phi \colon \mathcal{M}_1 \times \cdots \times \mathcal{M}_K \to \mathcal{X}^n$ that is a map  (stochastic in general) from the messages to channel inputs, i.e., $X^{[n]} = \phi(W_1,\cdots, W_K)$, and $K$ decoders,  $\psi_k  \colon \mathcal{Y}_k^n  \to \mathcal{M}_k, \forall k\in [K]$, (also stochastic maps in general) such that 
	$\widehat{W}_k \triangleq \psi_k(Y_k^{[n]})$ is the message decoded by Rx-$k$. The required independence of the stochastic encoding and decoding maps and the channel is specified by the form of the joint distribution of $W_1,\cdots, W_K, X^{[n]}, Y_1^{[n]}, \cdots, Y_K^{[n]}, \widehat{W}_1,\cdots, \widehat{W}_K$, which is expressed as,
\begin{align}
	& \Pr\left(W_{[K]}=w_{[K]},  X^{[n]}=x^{[n]}, Y_{[K]}^{[n]} = y_{[K]}^{[n]}, \widehat{W}_{[K]} = \widehat{w}_{[K]}\right) \notag \\
	& =\frac{1}{\prod_{k=1}^K |\mathcal{M}_k|}  \Pr\left(X^{[n]}=x^{[n]}\mid W_{[K]}=w_{[K]}\right)   ~  \mathcal{N}^{\otimes n}(y_{[K]}^{[n]}\mid x^{[n]})  \Pr(\widehat{W}_{[K]} = \hat{w}_{[K]} \mid Y_{[K]} = y_{[K]}^{[n]}). 
\end{align}
The probability of decoding error for Rx-$k$, and the overall probability of error for a coding scheme, are defined as, respectively,
\begin{align} \label{eq:def_error_probability_general}
	P_{e,k} \triangleq \Pr(\widehat{W}_k \not= W_k)
	, && P_e  \triangleq \max_{k\in [K]} P_{e,k}.
\end{align}
A rate tuple $(R_1,\cdots, R_K)\in \mathbb{R}_+^K$ is said to be \emph{achievable} by classical coding schemes if and only if there exists a sequence (indexed by $n$, shown as superscript $\cdot^{(n)}$) of classical coding schemes such that,
\begin{align}
	\lim_{n\to \infty} P_e^{(n)} = 0,  \label{eq:criteria1}\\
	~\mbox{and }\forall k\in [K],~~\lim_{n\to \infty} \log_2 \frac{|\mathcal{M}_k^{(n)}|}{n} \geq R_k. \label{eq:criteria2}
\end{align}
Note that condition \eqref{eq:criteria2} also implies that all $K$ limits on the LHS of \eqref{eq:criteria2} must exist. 
\begin{definition}[Classical capacity] \label{def:capacity_classical}
	The classical capacity region $\mathcal{C}(\mathcal{N})$ is defined as the closure of the set of all rate tuples achievable by classical coding schemes. In particular, the classical sum-capacity is defined as $C_\Sigma(\mathcal{N}) \triangleq \max_{(R_1,\cdots, R_K) \in \mathcal{C}(\mathcal{N})} (R_1+\cdots+R_K)$.
\end{definition}

\subsection{Non-Signaling Assistance: The NS box}
A $\kappa$-partite NS box $\mathcal{Z}$ with input $A_i \in\mathcal{A}_i$  and  output $B_i \in\mathcal{B}_i$ corresponding to the $i^{th}$ party,  $\forall i \in[\kappa]$, is specified by a conditional p.m.f. $\mathcal{Z} \colon  \mathcal{A}_1\times \cdots \times \mathcal{A}_\kappa \times \mathcal{B}_1 \times \cdots \times \mathcal{B}_\kappa \to [0,1]$,
\begin{align}
	&\mathcal{Z}(b_1,\cdots, b_\kappa \mid a_1,\cdots, a_\kappa)\notag\\
	&\triangleq \Pr(B_1=b_1,\cdots, B_\kappa =b_\kappa \mid A_1=a_1,\cdots, A_\kappa =a_\kappa).
\end{align}
We assume that the output alphabet sets $\mathcal{B}_1, \cdots, \mathcal{B}_\kappa$ have finite cardinality. The non-signaling condition \cite{barrett2005nonlocal, masanes2006general} requires that for all  $\{i_1,\cdots, i_m\} \subseteq [\kappa]$,
\begin{align} \label{eq:def_NS_condition}
	&\Pr(B_{i_1}=b_{i_1},\cdots, B_{i_m} = b_{i_m} \mid A_1=a_1,\cdots, A_\kappa=a_\kappa) \notag \\
	& =\Pr(B_{i_1}=b_{i_1},\cdots, B_{i_m} = b_{i_m} \mid A_{i_1}=a_{i_1},\cdots, A_{i_m} = a_{i_m}),
\end{align}
with the values of the variables chosen from their corresponding alphabets. In words, the condition says that the marginal distribution of the outputs of any subset of parties only depends on the inputs of those parties.\footnote{Intuitively, if this was not the case, then a subset of parties, by observing their own inputs and outputs, would be able to obtain some information about the inputs of the remaining parties, creating an opportunity for communication from the use of the NS box alone, in violation of the non-signaling principle.}
\begin{remark}Note that the NS box has classical inputs and outputs corresponding to each party, and those inputs and outputs are connected via the NS box. Any non-signaling resource, e.g., quantum entanglement that may be shared between the parties, must remain inside the box. This is the desired abstraction to study coding over classical channels. For \emph{quantum} channels (not considered in this work), one may indeed need a different definition of non-signaling assistance, such as  quantum NS assistance that is considered in \cite{QuantumNS}.
\end{remark}

\subsection{NS-assisted coding schemes} \label{sec:NSscheme}
An NS-assisted coding scheme operating over $n\in \mathbb{N}$ channel uses, utilizes a $\kappa =K+1$ partite NS box $\mathcal{Z}$ defined over 
$\mathcal{S} \times \mathcal{T}_1\times \cdots \times \mathcal{T}_K \times \mathcal{U} \times \mathcal{V}_1\times \cdots \times \mathcal{V}_K$, where 
$\mathcal{S} = \mathcal{M}_1\times \cdots \times \mathcal{M}_K$, 
$\mathcal{T}_k = \mathcal{Y}_k^n$ for all $k\in[K]$, 
$\mathcal{U} = \mathcal{X}^n$, and
$\mathcal{V}_k = \mathcal{M}_k$ for all $k\in [K]$.
The Tx is regarded as the $0^{th}$ party, corresponding to input $S\in \mathcal{S}$, and output $U\in \mathcal{U}$ of the NS box, while Rx-$k$, $k\in[K]$, is the $k^{th}$ party, with input $T_k\in \mathcal{T}_k$ and output $V_k\in \mathcal{V}_k$ of the NS box. The Tx sets the input of the NS box $S = (W_1,\cdots, W_K)$, and obtains the output $U$.
The Tx then sets $X^{[n]} = U$ as its transmitted sequence over the $n$ channel uses. Rx-$k$, $k\in[K]$, obtains $Y_k^{[n]}$ as the channel output, sets its input to the NS box as $T_k = Y_k^{[n]}$, and obtains the output $V_k$ from the box, as the decoded message $\widehat{W}_k$.

The joint distribution of $W_1,\cdots, W_K, X^{[n]}, Y_1^{[n]}, \cdots, Y_K^{[n]}, \widehat{W}_1,\cdots, \widehat{W}_K$ is expressed as, 
\begin{align} \label{eq:NS_joint_prob}
	& \Pr\left(W_{[K]}=w_{[K]},  X^{[n]}=x^{[n]}, Y_{[K]}^{[n]} = y_{[K]}^{[n]}, \widehat{W}_{[K]} = \widehat{w}_{[K]}\right) \notag \\
	&= \frac{1}{\prod_{k=1}^K |\mathcal{M}_k|} \mathcal{Z}\left(\left. x^{[n]},\widehat{w}_{[K]}~\right | ~w_{[K]}, y_{[K]}^{[n]}\right) ~ \mathcal{N}^{\otimes n}\left(y_{[K]}^{[n]} \mid x^{[n]}\right).
\end{align}
Probability of error is defined as in \eqref{eq:def_error_probability_general}. A rate tuple $(R_1,\cdots, R_K)\in \mathbb{R}_+^K$ is said to be achievable by NS-assisted coding schemes if and only if there exists a sequence (indexed by $n$) of NS-assisted coding schemes such that \eqref{eq:criteria1},\eqref{eq:criteria2} are satisfied.
\begin{definition}[NS-assisted capacity] \label{def:capacity_NS}
	The NS-assisted capacity region $\mathcal{C}^{\NS}(\mathcal{N})$ is defined as the closure of the set of all rate tuples achievable by NS-assisted coding schemes. In particular, the NS-assisted sum-capacity is defined as $C_\Sigma^{\NS}(\mathcal{N}) \triangleq \max_{(R_1,\cdots, R_K) \in \mathcal{C}^{\NS}(\mathcal{N})} (R_1+\cdots+R_K)$.
\end{definition}

\begin{remark}
There is no loss of generality in the framework presented above, because  all local processing operations carried out by each party can be absorbed into the NS box. This is because any processing done locally by a party is still non-signaling. For example, the framework allows the input to the channel to be a result of joint processing  of $U$ and $W_1,\cdots,$ $ W_K$, i.e., $X^{[n]} = \phi(U, W_1,\cdots, W_K)$ for some mapping $\phi$, and the decoding at Rx-$k$ to be a result of joint processing  of $V_k$ and $Y_k^{[n]}$, i.e., $\widehat{W}_k = \psi_k(V_k, Y_k^{[n]})$ for some mapping $\psi_k$, for $k\in [K]$.  As shown by \cite{allcock2009closed, beigi2015monotone}, the framework also allows  `\emph{wirings}' in which the parties may share multiple NS boxes, and let the input to a box be the output of other boxes in an arbitrary order. 
\end{remark}

It is worth noting that NS-assistance does not improve the capacity of a point to point channel.
\begin{lemma}[\!\!\cite{Bennett_Shor_Smolin_Thapliyal_PRL,matthews2012linear,Barman_Fawzi}\!\! ] \label{lem:p2p_capacity}
	The NS-assisted capacity of  a point to point discrete memoryless channel is equal to its classical capacity. 
\end{lemma}

\subsection{Coordinated Multipoint (CoMP) BC: $\mathbb{F}_q$ Model}\label{sec:Fqmodel}
Recall that we are interested in the wireless setting called Coordinated Multipoint Broadcast (CoMP BC).  The CoMP BC involves a transmitter with $B$ antennas (labeled Tx-$1$, $\cdots$, Tx-$B$), $K$ receivers (labeled Rx-$1$,$\cdots$, Rx-$K$), and a connectivity matrix ${\bf M}\in\{0,*\}^{K\times B}$ that specifies a fixed topology of the network. Rx-$i$ is connected to (has a non-zero channel coefficient to) Tx-$j$ if $M_{ij}=*$, and is \emph{not} connected to  (has a zero channel coefficient to) Tx-$j$ if $M_{ij}=0$, for all $(i,j)\in[K]\times[B]$.
For ease of exposition,\footnote{Noiseless finite field models are commonly employed to approximate the capacity of wireless networks \cite{Avestimehr_Diggavi_Tse, Jafar_TIM}, and are understood to be particularly meaningful in the high SNR limit, i.e., for DoF analyses.} we will primarily consider a finite field $(\mathbb{F}_q)$ model. In this model, over the $\tau^{th}$ channel use, Rx-$k$ obtains the channel output,
\begin{align}
\overline{Y}_k^{(\tau)}&=\left({Y}_k^{(\tau)},{\bf G}_k \triangleq (G_{kj}^{(\tau)})_{j\in [B]}\right),~~\forall k\in[K],\label{eq:Ykout}\\
\begin{bmatrix}
Y_1^{(\tau)}\\\vdots\\Y_K^{(\tau)}
\end{bmatrix}
&=
\begin{bmatrix}
		G_{11}^{(\tau)} & G_{12}^{(\tau)} & \cdots & G_{1B}^{(\tau)} \\
		G_{21}^{(\tau)} & G_{22}^{(\tau)} & \cdots & G_{2B}^{(\tau)} \\
		\vdots & \vdots & \ddots & \vdots \\
		G_{K1}^{(\tau)} & G_{K2}^{(\tau)} & \cdots & G_{KB}^{(\tau)} 
	\end{bmatrix}
\begin{bmatrix}
X_1^{(\tau)}\\\vdots\\X_{B}^{(\tau)}
\end{bmatrix}.\label{eq:def_channel_Fq}
\end{align}
Here $X_b^{(\tau)}\in\mathbb{F}_q$ is the signal sent from Tx-$b$, and $G_{ij}^{(\tau)}$ is the channel coefficient from Tx-$j$ to Rx-$i$. Over each channel use, the channel coefficients $G_{ij}^{(\tau)}$ are generated i.i.d. uniform from $\mathbb{F}_q^\times$ if $M_{ij}=*$, and held fixed at $0$ if $M_{ij}=0$. Note that including the channel coefficients in the output at Rx-$k$ as in \eqref{eq:Ykout} is simply a way to model perfect channel state information at the receivers (perfect CSIR). The channel connectivity matrix ${\bf M}$ remains fixed across channel uses and is globally known. Channel state information at the transmitter  is assumed unavailable (no-CSIT) beyond the fixed connectivity matrix ${\bf M}$, i.e., the random realizations of the non-zero coefficients are unknown to the transmitter. 

\begin{definition}[Fully Connected]\label{def:fullconnect}
A fully connected CoMP BC is one where  every Rx is connected to every Tx, i.e.,  $M_{ij}=*$ for all $i\in[K],j\in[B]$.
\end{definition}

\begin{definition}[Tree Network]\label{def:treenet}
A Tree Network is a CoMP BC with channel connectivity corresponding to a   \emph{rooted tree} graph $\mathcal{T}$. Apart from the root node (an imaginary Tx labeled Tx-$0$, assumed connected to every Rx) which serves only to orient the graph, there are $B$ vertices in $\mathcal{T}$, corresponding to Tx-$1$, $\cdots$, Tx-$B$. The defining condition of a tree network is that its connectivity matrix ${\bf M}$ must satisfy the following two properties: 
\begin{enumerate}
\item For every Rx-$k$, $k\in[K]$, the set of all Tx nodes to which it is connected, $\{\mbox{Tx-$j$}: M_{kj}=*\}$, comprise a \emph{path graph}, i.e., they are all on the \emph{same} path from the root-node. 
\item For every Rx-$k$, $k\in[K]$, if Rx-$k$ is connected to Tx-$j$, then Rx-$k$ must also be connected to all ancestors of Tx-$j$.
\end{enumerate}
Define depth(Tx-$j$) as the length of the path from the root node to Tx-$j$. For each Rx-$k$, define its `associated Tx', labeled Tx(Rx-$k$), as the one with the greatest depth among all Tx that are connected to Rx-$k$. Formally, 
\begin{align}
\mbox{Tx(Rx-$k$)}\triangleq \arg\max_{\mbox{\footnotesize Tx-$j$}: M_{kj}=*}\mbox{depth}(\mbox{Tx-$j$}).\label{eq:associate}
\end{align}
To avoid degenerate scenarios, we assume $B>0, K>0$,  that there is at least one Rx associated with each Tx-$b$, $\forall b\in[B]$, and there is at least one  (thus a unique)  Tx associated with each Rx-$k$, $\forall k\in[K]$. The number of leaf nodes of $\mathcal{T}$ is denoted as $\ell(\mathcal{T})$. A tree graph $\mathcal{T}$ is called a `\emph{path graph}' if it has only one leaf node, $\ell(\mathcal{T})=1$.
\end{definition}

Recall the classical and NS-assisted capacity regions in Definition \ref{def:capacity_classical} and Definition \ref{def:capacity_NS}.
Let $\mathcal{C}(q), \mathcal{C}^{\NS}(q)$ denote the classical and NS-assisted capacity regions for the CoMP BC $\mathbb{F}_q$ model, respectively.
The exact $\mathcal{C}(q)$ may be intractable. To obtain a meaningful approximation, we invoke the notion of degrees of freedom (DoF), studied widely in wireless networks, translated to the $\mathbb{F}_q$ model. Intuitively, as the ratio of the sum-capacity of a communication network to the capacity of a single interference-free point-to-point channel, DoF represent the number of  interference-free ($q$-ary) channels that can be created in the network for the desired messages, as $q$ approaches infinity. Formally, a DoF tuple $$(d_1,d_2,\cdots, d_K)\in \mathbb{R}_+^K$$ is said to be achievable by classical/NS-assisted coding schemes if and only if for all $q\geq 2$, $$\exists (R_1(q),R_2(q),\cdots, R_K(q))\in \mathcal{C}^{\#}(q)$$ such that
\begin{align}
	\lim_{q\to \infty} \frac{R_k(q)}{\log_2 q} \geq  d_k, ~~ \forall k\in [K],\label{eq:limq}
\end{align}
where $\#$ is a placeholder that may be replaced with `NS' if NS-assistance is allowed. The limit $q\rightarrow\infty$ in \eqref{eq:limq} is defined over the sequence of \emph{all} feasible $q$ values, i.e., all natural numbers that can be expressed as powers of prime numbers, arranged in ascending order, i.e., $2,3,4,5,7,8,9,11,$ $13, 16,\cdots$.

The classical DoF region $\mathcal{D}$ is defined as the closure of all DoF tuples achievable by classical coding schemes. The NS-assisted DoF region $\mathcal{D}^{\NS}$ is defined as the closure of all DoF tuples achievable by NS-assisted coding schemes. In particular, the classical sum-DoF  $d_{\Sigma}$, the NS-assisted sum-DoF $d_{\Sigma}^{\NS}$ are defined as $d_{\Sigma} \triangleq \max_{(d_1,\cdots, d_K)\in \mathcal{D}}(d_1+\cdots +d_K)$, $d_{\Sigma}^{\NS} \triangleq \max_{(d_1,\cdots, d_K)\in \mathcal{D}^{\NS}}(d_1+\cdots +d_K)$, respectively.

\subsection{CoMP BC: Gaussian model}\label{sec:gaussianmodel}
The Gaussian model is similar to the $\mathbb{F}_q$ model, except the symbols and operations are over $\mathbb{R}$ instead of $\mathbb{F}_q$, there is additive Gaussian noise at the receivers, and the channel inputs are subject to a transmit power constraint. Over the $\tau^{th}$ channel use, Rx-$k$ obtains the channel output,
\begin{align}
\overline{Y}_k^{(\tau)}&=\left({Y}_k^{(\tau)},{\bf G}_k \triangleq (G_{kj}^{(\tau)})_{j\in [K]}\right),~~\forall k\in[K],\label{eq:YkoutG}\\
\begin{bmatrix}\label{eq:def_channel_Gaussian}
Y_1^{(\tau)}\\\vdots\\Y_K^{(\tau)}
\end{bmatrix}
&=
\begin{bmatrix}
		G_{11}^{(\tau)} & G_{12}^{(\tau)} & \cdots & G_{1K}^{(\tau)} \\
		G_{21}^{(\tau)} & G_{22}^{(\tau)} & \cdots & G_{2K}^{(\tau)} \\
		\vdots & \vdots & \ddots & \vdots \\
		G_{K1}^{(\tau)} & G_{K2}^{(\tau)} & \cdots & G_{KK}^{(\tau)} 
\end{bmatrix}
\begin{bmatrix}
X_1^{(\tau)}\\\vdots\\X_{K}^{(\tau)}
\end{bmatrix}
+
\begin{bmatrix}
Z_1^{(\tau)}\\\vdots\\Z_K^{(\tau)}
\end{bmatrix}.
\end{align}
Here $X_k^{(\tau)}\in\mathbb{R}$ is the signal sent from Tx-$k$.   $Z_k^{(\tau)} \sim \mathcal{N}(0,1)$ are i.i.d. Gaussian noise terms with zero mean and unit variance. $G_{ij}^{(\tau)}$ is the channel coefficient from Tx-$j$ to Rx-$i$, and is held fixed at $0$ if $M_{ij}=0$. For each connected link, i.e., $(k,j)$ such that $M_{kj} =\ast$, the channel coefficient values are bounded away from zero and infinity, $1/c\leq |G_{kj}^{(\tau)}| \leq c$ for a positive constant $c$, and are generated i.i.d. according to a probability density function $f(G)$ whose peak value is bounded by some constant, i.e., $\sup f(\cdot)=f_{\max}<\infty$.

For coding schemes spanning $n$ channel uses, the inputs must satisfy the transmit power constraint,
\begin{align} \label{eq:power_constraint}
	\mathbb{E}\Big[ \frac{1}{n} \sum_{\tau=1}^n \Big( |X_1^{(\tau)}|^2 + \cdots + |X_{K}^{(\tau)}|^2 \Big) \Big] \leq P,
\end{align}
i.e., the average transmit power is upper bounded by $P$. 

Let $\mathcal{C}(P), \mathcal{C}^{\NS}(P)$ denote the classical and NS-assisted capacity regions with respect to the average transmit power constraint $P$. A degree of freedom (DoF) tuple $(d_1,d_2,\cdots, d_K)\in \mathbb{R}_+^K$ is  achievable by classical/NS-assisted coding schemes if and only if for all $P>0$, $\exists (R_1(P),\cdots, R_K(P))\in \mathcal{C}^{\#}(P)$ such that 
	$\lim_{P \to \infty} \frac{R_k(P)}{\frac{1}{2}\log_2 P} \geq  d_k, ~~ \forall k\in [K]$,
where $\#$ is a placeholder that may be replaced with `NS' if NS-assistance is allowed. The classical DoF region $\mathcal{D}$ is defined as the closure of all DoF tuples achievable by classical coding schemes. The NS-assisted DoF region $\mathcal{D}^{\NS}$ is defined as the closure of all DoF tuples achievable by NS-assisted coding schemes. In particular, the classical sum-DoF  $d_{\Sigma}$, the NS-assisted sum-DoF $d_{\Sigma}^{\NS}$ are defined as $d_{\Sigma}  \triangleq \max_{(d_1,\cdots, d_K)\in \mathcal{D}}(d_1+\cdots +d_K)$, $d_{\Sigma}^{\NS} \triangleq \max_{(d_1,\cdots, d_K)\in \mathcal{D}^{\NS}}(d_1+\cdots +d_K)$, respectively.

\section{Results}
\subsection{Fully Connected CoMP BC}
Our first result, stated in the following theorem, shows that NS-assistance does not improve the DoF region, or even the capacity region, in the fully-connected (Definition \ref{def:fullconnect}) CoMP BC.
\begin{theorem}[Fully Connected CoMP BC Capacity and DoF Regions]\label{thm:fullconnect} 
For a fully-connected CoMP BC network, the capacity regions with and without NS-assistance are characterized, under the $\mathbb{F}_q$ model as,
\begin{align}
\mathcal{C}^{\NS}(q)=\mathcal{C}(q)= \left\{ (R_1,\cdots, R_K)  \in \mathbb{R}_+^K ~\left|	~~	R_1+R_2+\cdots+R_K\leq C_1(q)\right.\right\},
\end{align}
and under the Gaussian model as,
\begin{align}
\mathcal{C}^{\NS}(P)=\mathcal{C}(P)= \left\{ (R_1,\cdots, R_K)  \in \mathbb{R}_+^K ~\left|	~~	R_1+R_2+\cdots+R_K\leq C_1(P)\right.\right\},
\end{align}
where $C_1(q), C_1(P)$ represent the single-user capacity of Rx-$1$ under the two models. Note that $C_1(q)=\log_2(q)$, and $C_1(P) = \max_{{\bf Q}}\mathbb{E}\big[ \frac{1}{2} \log_2 (1+{\bf G}_1 {\bf Q} {\bf G}_1^\top) \big]$, where the maximization is over all positive semi-definite ${\bf Q}$ with ${\rm Tr}({\bf Q}) = P$ \cite{goldsmith2003capacity}. The corresponding DoF regions, under both the $\mathbb{F}_q$ model and the Gaussian model, are characterized as,
\begin{align}
\mathcal{D}^{\NS}=\mathcal{D}= \left\{ (d_1,d_2,\cdots, d_K)  \in \mathbb{R}_+^K ~\left|	~~	d_1+d_2+\cdots+d_K\leq 1\right.\right\}.
\end{align}
\end{theorem}
The key to Theorem \ref{thm:fullconnect} is the \emph{statistical-equivalence}, or the \emph{same-marginals} property of the receivers, which makes them indistinguishable from the transmitter's perspective. The \emph{same-marginals} argument \cite{sato1978outer,Dimacs_BC} is a standard line of reasoning in classical literature that makes use of the fact that the probabilities of error experienced by the receivers for an arbitrary (classical) coding scheme  depend only on the marginal channel distribution of each receiver. In our fully connected CoMP BC since the marginal distributions are identical across receivers, the same-marginals property ensures that the capacity and DoF regions remain unchanged if every receiver has exactly the same channel realizations as Rx-$1$, in every channel-use. Once all receivers observe the same channel output, even allowing full cooperation among the receivers cannot change the capacity or DoF regions. Therefore the sum-capacity cannot exceed the single-user capacity, and any allocation of rates across messages that does not exceed the single-user capacity is trivially achievable, implying immediately the classical capacity and DoF regions in Theorem \ref{thm:fullconnect}. Beyond the classical case, in order to show that NS-assistance cannot improve the capacity and DoF regions, two additional facts are needed. 
\begin{enumerate}[align=left]
\item[Fact 1:] For a point to point (single user) channel, NS-assistance cannot improve the capacity. Fortunately, this non-trivial fact is already well-established, as noted in Lemma \ref{lem:p2p_capacity} \cite{Bennett_Shor_Smolin_Thapliyal_PRL,matthews2012linear,Barman_Fawzi}.
\item[Fact 2:] The same-marginals property still holds under NS-assistance, even when NS-assistance is available to all parties.\footnote{The same-marginals property has been shown in \cite{fawzi2024broadcast}  for a BC with NS-assistance available to only the decoders, i.e., receivers, when the metric of interest is the sum of probabilities of error of the receivers.}
\end{enumerate}
With these two facts, the proof of Theorem \ref{thm:fullconnect} is straightforward in the NS-assisted setting (Use Fact 2 to make the channels identical across receivers, reduce to single receiver by allowing cooperation among receivers, then use Fact 1). So it only remains to establish Fact 2. This is done in the following theorem, not just for CoMP BC, but for the general BC setting of Section \ref{sec:dmcbc}.

\begin{theorem}[Same-marginals Property] \label{thm:same_marginal}
	Given two $K$-user BCs $\mathcal{N}_{Y_1\cdots Y_K\mid X}$ and $\widetilde{\mathcal{N}}_{Y_1\cdots Y_K\mid X}$, if their marginal distributions are the same, i.e., $\mathcal{N}_{Y_k\mid X}=\widetilde{\mathcal{N}}_{Y_k\mid X}$ for all $k\in [K]$, then for any NS-assisted (or classical) coding scheme, $P_{e,k} = \widetilde{P}_{e,k}$, where $P_{e,k},\widetilde{P}_{e,k}$ denote the probability of error of the $k^{th}$ message.
\end{theorem}
The proof of Theorem \ref{thm:same_marginal} is provided in Appendix \ref{sec:proofsamemarginal}. Compared to \cite{fawzi2024broadcast}, note that Theorem \ref{thm:same_marginal} applies to each user's error probability, and more importantly, holds even with NS-assistance to all parties (including the transmitter).  

\subsection{Semi-deterministic BC}
Before proceeding to more general classes of CoMP BCs in the subsequent sections, let us consider a $K=2$ user CoMP BC setting in this section. This simple $K=2$ setting turns out to be  interesting because 1) it answers the open question of \cite{fawzi2024broadcast} by demonstrating a strict sum-capacity advantage due to NS-assistance in a BC, 2) it is a semi-deterministic BC, thus proving that NS-assistance improves capacity even in a semi-deterministic BC (recall that NS-assistance does not improve capacity in a deterministic BC \cite{fawzi2024broadcast}), and 3) it shows that the improvement in sum-capacity  due to NS-assistance can be as high as a factor of $2$, which is also the largest possible value.

Let us first provide the formal definition of a semi-deterministic BC.
\begin{definition}[Semi-deterministic BC]
	A 2-user broadcast channel $(\mathcal{X}, (\mathcal{Y}_k)_{k\in [2]}, \mathcal{N}_{Y_1Y_2\mid X})$ is semi-deterministic (cf. \cite[Sec. 8.3.1]{NIT}) if and only if there exists $k\in \{1,2\}$ such that $Y_k$ is determined by $X$, i.e., the marginal distribution $\mathcal{N}_{Y_k\mid X}(y_k\mid x) \in \{0,1\}$ for all $x\in \mathcal{X}, y_k \in \mathcal{Y}_k$.
\end{definition}

\begin{theorem}\label{thm:semidet}
Let $\mathfrak{N}^{\semidet}$ be the set of all ($2$-user) semi-deterministic broadcast channels, and $C_\Sigma^{\NS}(\mathcal{N})$, $C_\Sigma(\mathcal{N})$ the sum-capacity of $\mathcal{N}\in\mathfrak{N}^{\semidet}$ with and without NS-assistance, respectively. Then,
\begin{align}
\sup_{\mathcal{N}\in\mathfrak{N}^{\semidet}}\frac{C_\Sigma^{\NS}(\mathcal{N})}{C_{\Sigma}(\mathcal{N})}=2.
\end{align}
\end{theorem}

\begin{proof}
	A gain by a factor larger than $2$ is impossible for any $2$-user broadcast channel because for any rate tuple $(R_1,R_2)$ that is achievable by an NS-assisted coding scheme, the rate tuples $(R_1,0)$  and $(0, R_2)$ are achievable by classical coding schemes (Lemma \ref{lem:p2p_capacity}). Therefore the gain is upper bounded by $(R_1+R_2)/\max(R_1,R_2) \leq 2$. For the other direction, let us provide a toy example of a semi-deterministic BC $(\mathcal{N}^{\mbox{\tiny toy1}})$, illustrated in Fig. \ref{fig:NS_2user_bipartite}, for which the NS-assisted sum-capacity will be shown to be exactly $C_\Sigma^{\NS}(\mathcal{N}^{\mbox{\tiny toy1}})=2\log_2 q$, while the sum-capacity without NS-assistance will be shown to be $C_\Sigma(\mathcal{N}^{\mbox{\tiny toy1}})=\log_2 q+o_q(\log_2 q)$. Thus, the ratio ${C_\Sigma^{\NS}(\mathcal{N}^{\mbox{\tiny toy1}})}/{C_{\Sigma}(\mathcal{N}^{\mbox{\tiny toy1}})}\rightarrow 2$ as $q\rightarrow\infty$,  proving the other direction.

\begin{definition}[Toy Channel $\mathcal{N}^{\mbox{\tiny toy1}}_{\overline{Y}_1,\overline{Y}_2\mid X_1,X_2}$]\label{def:toy1}
The toy channel $\mathcal{N}^{\mbox{\tiny \rm toy1}}_{\overline{Y}_1,\overline{Y}_2\mid X_1,X_2}$ corresponds to a CoMP BC ($\mathbb{F}_q$ model, $q>2$) with $K=2$ users, as shown in Figure \ref{fig:NS_2user_bipartite}. Note that unlike the CoMP BC model which would have $G_{11},G_{22},G_{21}$ i.i.d. uniform in $\mathbb{F}_q^\times$,  the toy channel fixes channel coefficients $G_{11}=G_{22}=1$. This is without loss of generality, because even if the coefficients were randomly drawn in $\mathbb{F}_q^\times$, each receiver Rx-$k$, $k\in\{1,2\}$,  can normalize its received signal by $1/G_{kk}$. This operation is reversible so it does not change the capacity, but it brings the channel to the form that $G_{11}=G_{22}=1$ while $G_{21}$ remains uniform in $\mathbb{F}_q^\times$. The resulting channel $\mathcal{N}^{\mbox{\tiny \rm toy1}}_{\overline{Y}_1,\overline{Y}_2\mid X_1,X_2}$ is indeed semi-deterministic, as the channel to Rx-$1$ is now deterministic, $\overline{Y}_1=Y_1 = X_1$, and the channel to Rx-$2$ is $\overline{Y}_2 = (Y_2,G)= (GX_1+X_2, G)$.
\end{definition}
\begin{figure}[!h]
\center
\begin{tikzpicture}
\node (myfirstpic) at (0,0) {\includegraphics[width=0.34\textwidth]{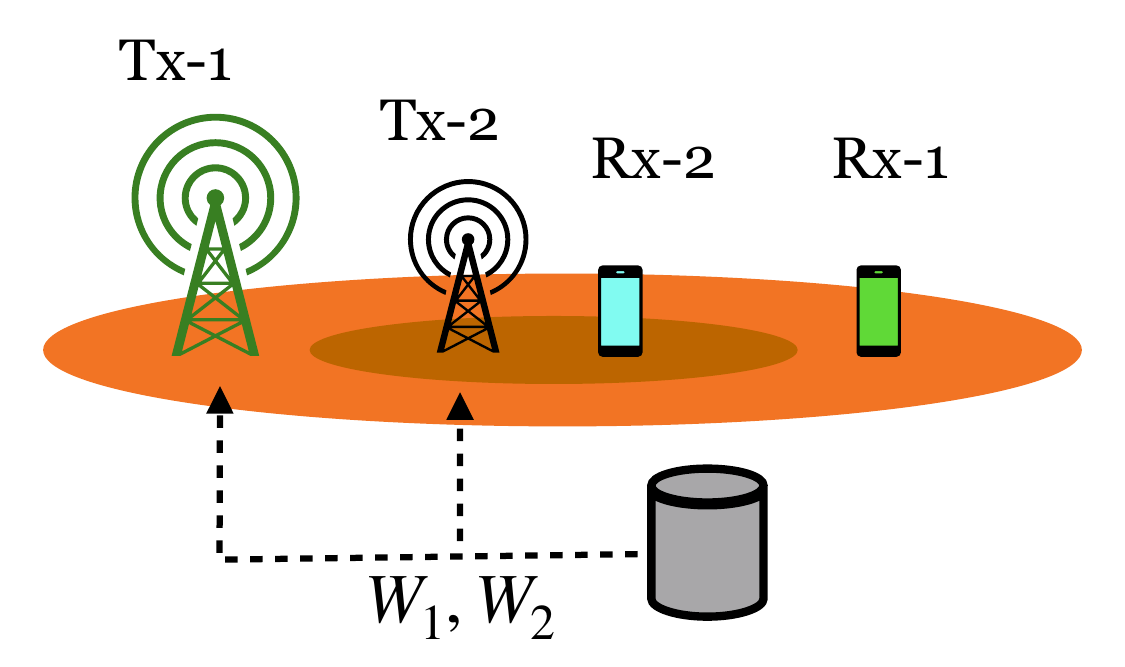}};
\node at (5,0){\(
   \kbordermatrix{
    & \mbox{\tiny Tx-$1$} & \mbox{\tiny Tx-$2$}  \\
    \mbox{\tiny Rx-$1$} & * & 0   \\
    \mbox{\tiny Rx-$2$} & * & *   \\
  }={\bf M}
  \)};
  \end{tikzpicture}\\[0.3cm]
\begin{tikzpicture}
\begin{scope}[shift = {(0,1)}]
\node[rectangle, minimum width = 1cm, minimum height = 3.6cm, fill = gray!20] (Tx) at (-2,0.81) {};
\node (X1) at (-2,1.5) {\small $X_1$};
\node (X2) at (-2,-0.5) {$X_2$};
\node [fill = black!30, rectangle, above=0.1cm of X1, minimum width = 0.85cm] {\small Tx};
\node[rectangle, minimum width = 1cm, minimum height = 1.5cm, fill = gray!20] (Rx1) at (2,1.75) {};
\node (Y1) at (2,1.5) {$Y_1$};
\node [fill = black!30, rectangle, above=0.1cm of Y1, minimum width = 0.85cm] {\small Rx-$1$};
\node[rectangle, minimum width = 1cm, minimum height = 2.4cm, fill = gray!20] (Rx2) at (2,-0.7) {};
\node (Y2) at (2,-0.5) {$Y_2$};
\node [fill = black!30, rectangle, above=0.1cm of Y2, minimum width = 0.85cm] {\small Rx-$2$};
\node (G) at (2,-1.5) {$G$};

\draw [line width = 1] (X1.east) -- (Y1.west);
\draw [line width = 1] (X1.east) -- (Y2.west) node [above = 0cm, pos= 0.55] {$G$};
\draw [line width = 1] (X2.east) -- (Y2.west);

\node[draw, rectangle, thick, minimum width = 5.7cm, minimum height = 6.3cm] (Channel) at (0,-0.2) {};
\node at (0,2.3) [align=center]{Channel\\ $\mathcal{N}^{\mbox{\tiny toy1}}_{\overline{Y}_1,\overline{Y}_2\mid X_1,X_2}$};
\end{scope}

\node [align = left,blue] at (-0.7,-1.2) {
\small \scalebox{0.9}{$G\sim {\rm Unif}(\mathbb{F}_q^{\times})\indep(X_1,X_2)$}\\
\small \scalebox{0.9}{\colorbox{white}{$\overline{Y}_1=Y_1 = X_1$}}\\[0.1cm]
\small \scalebox{0.9}{\colorbox{white}{$\overline{Y}_2=\begin{bmatrix}Y_2\\ G\end{bmatrix} = \begin{bmatrix}GX_1+X_2\\ G\end{bmatrix}$}}
};

\begin{scope}[shift = {(0,-5)}]
\node[rectangle, minimum width = 8cm, minimum height = 2cm, fill = gray!30] at (0,0) {};
\node at (0,1.4){Non-signaling Box $\mathcal{Z}$};
\node at (0,0)[align=center]{$\mathcal{Z}(u, v \mid s,t) =
	\begin{cases}
		1/q, & u+v = s\cdot t\\
		0, & u+v \not= s\cdot t
	\end{cases}$ \\[0.1cm]
	$\forall s,t, u,v\in \mathbb{F}_q$};
\node[rectangle, minimum width = 8cm, minimum height = 3.5cm, draw, thick] (NS) at (0,0) {};
\node (S) at (-3.25,1.35) {$s$};
\node (T) at (3.25,1.35) {$t$};
\node (U) at (-3.25,-1.35) {$u$};
\node (V) at (3.25,-1.35) {$v$};

\node [blue] at (0,-2.2) {$U+V = GW_1$};

\end{scope}

\node [draw, rectangle, thick, left = 1.5cm of X1,fill=red!10,  minimum width = 2cm] (W1) {$W_1 \in \mathbb{F}_q$}; 

\node [draw, rectangle, thick, fill=blue!10, below = 0.2cm of W1, minimum width = 2cm] (W2) {$W_2 \in \mathbb{F}_q$};  

\node [draw, rectangle, thick, left = 1.57cm of X2, minimum width = 1.8cm, minimum height = 0.8cm] (W2U) {$W_2-U$};

\node [draw, rectangle, thick, minimum height = 1.5cm, align=left] (Y2V) at (5.75,0) {$Y_2-V$};
\node (Y1W1) [right = 2.5cm of Y1] {$Y_1 =X_1= \colorbox{red!10}{$W_1$}$};

\node (W2out) [right = 0.5cm of Y2V, fill=blue!10] {$W_2$};

\draw [-latex, thick] (W1.east) -- ($(X1.west)+(-0.15,0)$);
\draw [-latex, thick] (W2.south) -- ($(W2.south)+(0,-0.4)$);
\draw [-latex, thick] (W2U.east) -- ($(X2.west)+(-0.15,0)$);
\draw [-latex, thick] (W1.west) -- ($(W1.west)+(-0.5,0)$) -- ($(W1.west)+(-0.5,-4)$) -- ($(W1.west)+(2.65,-4)$)  -- ($(W1.west)+(2.65,-5.75)$) node [left, pos = 0.5] {$W_1$};
\draw [-latex, thick] ($(Y1.east)+(0.2,0)$) -- (Y1W1.west);
\draw [-latex, thick] ($(G.east)+(0.2,0)$) -- ($(G.east)+(1,0)$) -- ($(G.east)+(1,-2.75)$) node [pos=0.8,right] {$G$};
\draw [-latex, thick] ($(Y2.east)+(0.2,0)$) -- ($(Y2.east)+(2.7,0)$);

\draw [-latex, thick] ($(V.south)+(0,-0.18)$) --  ($(V.south)+(0,-0.7)$) --  ($(V.south)+(1.2,-0.7)$)  -- ($(V.south)+(1.2,6.1)$) node [pos= 0.07, right=0.4cm, rotate=90, blue]{$V=(GW_1-U)\sim\mbox{Unif}(\mathbb{F}_q)$} -- ($(V.south)+(1.78,6.1)$);

\draw [thick] ($(U.south)+(0,-0.18)$) -- ($(U.south)+(0,-0.7)$) -- ($(U.south)+(-1.65,-0.7)$) -- ($(U.south)+(-1.65,1+3.89)$) node [pos= 0.9, left=0.4cm, rotate=90, blue]{$U\sim\mbox{Unif}(\mathbb{F}_q)\indep(W_1,G)$} ;

\draw [-latex, thick] ($(U.south)+(-1.65,1+4.29)$) -- ($(U.south)+(-1.65,1+5.7)$);

\draw [thick] (-4.89,1-2.29) arc (90:270:0.2);

\draw [-latex, thick] (Y2V) -- (W2out);

\end{tikzpicture}
\caption{$K=2$ toy example, with a capacity achieving coding scheme based on a bipartite NS box. A general description of the coding scheme for arbitrary $K$ appears in Appendix \ref{sec:altproof}. Note that in this coding scheme NS-assistance is not utilized by Rx-$1$ or Tx-$1$ (the first transmit antenna). Essentially, the NS advantage in this example can be traced down to the \emph{Fading Dirty Paper} channel (Section \ref{sec:fadingdirt}) between Tx-$2$ and Rx-$2$ where Tx-$2$ has the knowledge of the `\emph{dirt}' ($X_1$) while Rx-$2$ has the knowledge of the `\emph{fading}' ($G$).}
\label{fig:NS_2user_bipartite}
\end{figure}
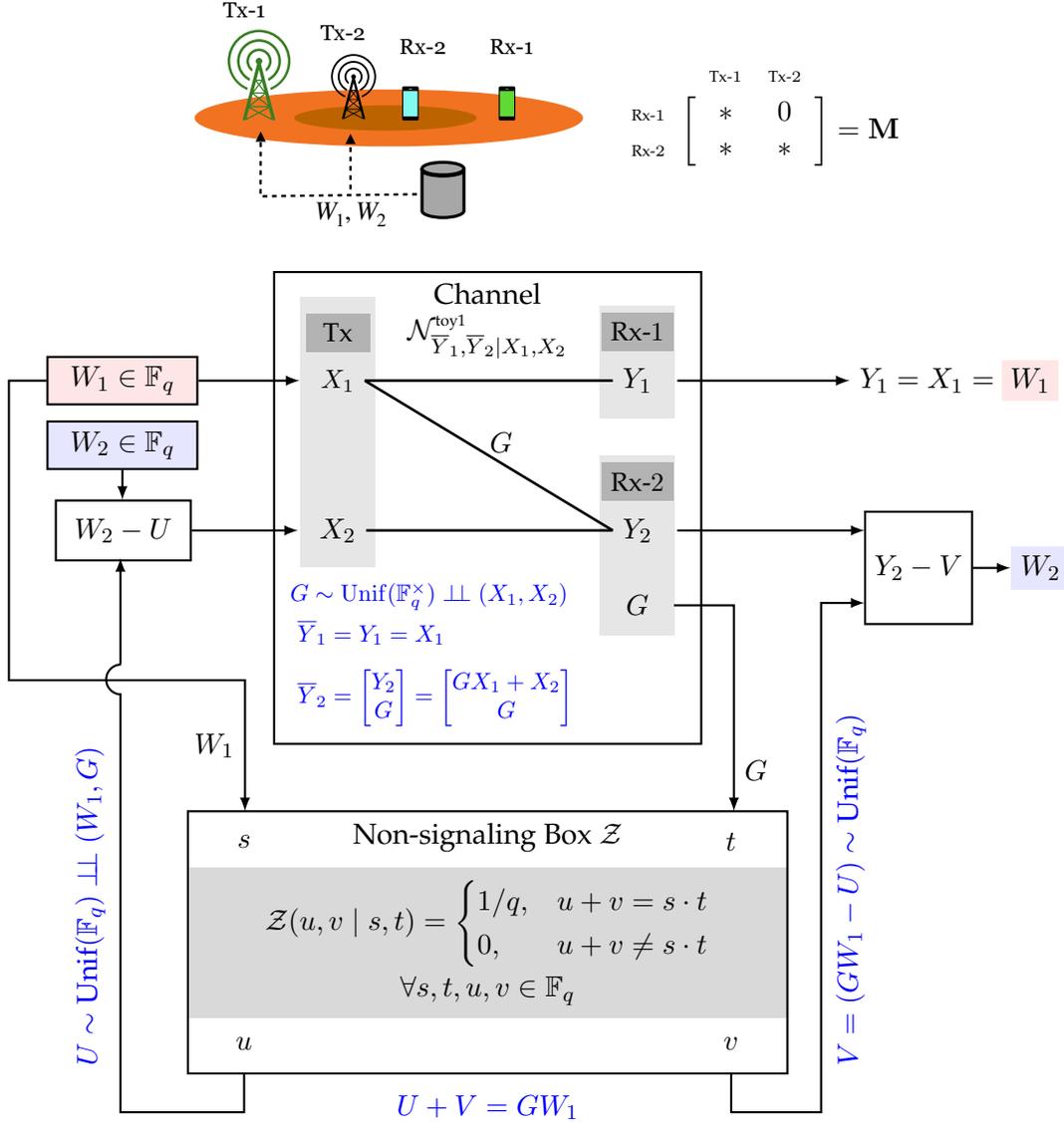

On this toy channel $\mathcal{N}^{\mbox{\tiny toy1}}_{\overline{Y}_1,\overline{Y}_2\mid X_1,X_2}$, the NS-assisted coding scheme shown in Fig. \ref{fig:NS_2user_bipartite}  achieves $(R_1,R_2) = (\log_2 q, \log_2 q)$ by sending one $q$-ary symbol for each message ($W_1,W_2$) over one channel-use, which can be recovered by the corresponding receivers with vanishing (in fact exactly zero) error probability. The scheme works as follows. The Tx sends $W_1$ to its input of the non-signaling box, and obtains the output $U$. For its inputs of the channel, it sets $X_1=W_1$, $X_2 = W_2-U$.  Rx-$1$ obtains $Y_1=W_1$ directly from its channel. Rx-$2$ obtains $Y_2=GX_1+X_2 = GW_1+W_2-U$ and $G$ from its channel. It sends $G$ to its input  of the non-signaling box, and obtains the output $V$. The non-signaling box ensures the relationship $U+V = GW_1$. Rx-$2$ finally subtracts $V$ from $Y_2$ and obtains $GW_1+W_2-(U+V) = W_2$. Thus, $C_{\Sigma}^{\NS}(\mathcal{N}^{\mbox{\tiny toy1}})\geq 2\log_2 q$. In fact $C_{\Sigma}^{\NS}(\mathcal{N}^{\mbox{\tiny toy1}})= 2\log_2 q$ because NS-assisted capacity is bounded above by the capacity with full cooperation among the two receivers (Lemma \ref{lem:p2p_capacity}), which is itself bounded as $I(X_1,X_2; Y_1,Y_2,G)\leq H(Y_1,Y_2|G)\leq 2\log_2 q$. 

To complete the proof of Theorem \ref{thm:semidet}, it remains to show that the classical capacity is upper bounded as $C_\Sigma(\mathcal{N}^{\mbox{\tiny toy1}})\leq \log_2(q) + o_q(\log_2 q)$. Note that a semi-deterministic BC is one of the few BC settings for which the capacity region is known \cite{semidetbc} \cite[Section 8.3.1]{NIT}. For our example, the region\footnote{Let us note that the auxiliary $U$ in \eqref{eq:toycapregion} is not related to the output of the NS-box labeled $U$ in the NS-assisted coding scheme. The overloading of $U$ should not cause confusion since the two contexts (NS versus classical) do not overlap.} is,
\begin{align}
\mathcal{C}(\mathcal{N}^{\mbox{\tiny toy1}})&=\bigcup_{P_{U,X_1,X_2}}\left\{\begin{array}{ll}(R_1,R_2):\\
R_1\leq H(\overline{Y}_1)\\
R_2\leq I(U;\overline{Y}_2)\\
R_1+R_2\leq H(\overline{Y}_1\mid U)+I(U;\overline{Y}_2)
\end{array}\right\}.\label{eq:toycapregion}
\end{align}
However, this characterization does not reveal the asymptotic value of the sum-capacity in the large alphabet limit  $q\rightarrow\infty$, which is what we need to produce the sum-capacity bound $C_\Sigma(\mathcal{N}^{\mbox{\tiny toy1}})\leq \log_2(q) + o_q(\log_2 q)$ for Theorem \ref{thm:semidet}. The desired bound instead follows from the argument presented in Appendix \ref{sec:AIS} which adapts the Aligned Images bound of \cite{davoodi2016aligned} to our finite field ($\mathbb{F}_q$, large $q$) model.
\end{proof}

Theorem \ref{thm:semidet} shows via a toy example how the strongest possible (factor of $2$) advantage of NS-assistance can be achieved in the large alphabet limit $(q\rightarrow\infty)$. A natural question one might ask is what happens, say in the same toy example, for small alphabet. The NS-assisted scheme shown in Fig. \ref{fig:NS_2user_bipartite} works over any finite field, so we always have $C_{\Sigma}^{\NS}(\mathcal{N}^{\mbox{\tiny toy1}})= 2\log_2 q$. So the question boils down to the classical sum-capacity. Let us consider this question. Since the channel is defined over a finite field, the smallest possible  setting for $\mathcal{N}^{\mbox{\tiny toy1}}_{\overline{Y}_1,\overline{Y}_2\mid X_1,X_2}$ is $\mathbb{F}_2$. But in this case since $\mathbb{F}_q^\times=\{1\}$, we have $G=1$ as a constant, and the knowledge of the channel connectivity matrix constitutes perfect CSIT. It is easily seen that in this case there is no advantage from NS-assistance, i.e., $C_\Sigma(\mathcal{N}^{\mbox{\tiny toy1}})=C_{\Sigma}^{\NS}(\mathcal{N}^{\mbox{\tiny toy1}})= 2\log_2 q$, achieved classically by setting $X_1=W_1, X_2=W_2-W_1$, which produces $Y_1=W_1, Y_2=W_1+W_2-W_1=W_2$. Thus, the smallest non-trivial alphabet corresponds to $q=3$, i.e., the finite field $\mathbb{F}_3$. 
The next theorem shows that NS-assistance still provides a significant sum-capacity advantage in this small alphabet setting, albeit much smaller than the factor of $2$ that was shown for large alphabet.

\begin{theorem}\label{thm:semidetF3}
For the semi-deterministic BC $\mathcal{N}^{\mbox{\tiny \rm toy1}}_{\overline{Y}_1,\overline{Y}_2\mid X_1,X_2}$ with $q=3$, i.e., over the finite field $\mathbb{F}_3$,
\begin{align}
C_\Sigma^{\NS}(\mathcal{N}^{\mbox{\tiny \rm toy1}})=2\log_2(3),&& C_{\Sigma}(\mathcal{N}^{\mbox{\tiny \rm toy1}})=1.5\log_2(3),
\end{align}
and therefore, ${C_\Sigma^{\NS}(\mathcal{N}^{\mbox{\tiny \rm toy1}})}/{C_{\Sigma}(\mathcal{N}^{\mbox{\tiny \rm toy1}})}=4/3$.
\end{theorem}
\proof As noted above, the NS-assisted coding scheme works for any $\mathbb{F}_q$,  so $C_\Sigma^{\NS}(\mathcal{N}^{\mbox{\tiny toy1}})=2\log_2(3)$. To show that $C_{\Sigma}(\mathcal{N}^{\mbox{\tiny toy1}})\geq 1.5\log_2(3)$, let us set in the capacity region expression \eqref{eq:toycapregion}, $(X_1,X_2,U)=(W_1, W_2-W_1,W_2)$ where $W_1,W_2$ are two i.i.d. uniform random variables in $\mathbb{F}_3$. This produces the bounds 
\begin{align}
R_1&\leq H(\overline{Y}_1)=H(W_1)=\log_2(3),\\
R_2&\leq I(U;\overline{Y}_2)=I(U;GX_1+X_2\mid G)\notag\\
&=\Pr(G=1)I(W_2;W_2)+\Pr(G=-1)I(W_2;W_1+W_2)\notag\\
&=0.5\log_2(3)\\
R_1+R_2&\leq H(\overline{Y}_1\mid U)+I(U;\overline{Y}_2)\notag\\
&=H(W_1\mid W_2)+0.5\log_2(3)\notag\\
&=1.5\log_2(3)
\end{align}
This establishes the achievability of the tuple $(R_1,R_2)=(\log_2(3),0.5\log_2(3))$ which satisfies all three bounds. Thus, we have the lower bound $C_{\Sigma}(\mathcal{N}^{\mbox{\tiny toy1}})\geq 1.5\log_2(3)$. Note that we represent $\mathbb{F}_3^\times=\{1,-1\}$, so $\Pr(G=1)=\Pr(G=-1)=0.5$. Next, again starting from the capacity region \eqref{eq:toycapregion} the converse bound $C_{\Sigma}(\mathcal{N}^{\mbox{\tiny toy1}})\leq 1.5\log_2(3)$ is shown as follows.
\begin{align}
&R_1+R_2\notag\\
&\leq H(\overline{Y}_1\mid U)+I(U;\overline{Y}_2) \\
&=H(\overline{Y}_1\mid U)+I(U;GX_1+X_2\mid G)\label{eq:useGprop1} \\
&=H(\overline{Y}_1\mid U)+0.5I(U;X_2+X_1)+0.5I(U;X_2-X_1)\label{eq:useGprop2}\\
&=H(\overline{Y}_1\mid U)+0.5H(X_2+X_1)-0.5H(X_2+X_1\mid U)+0.5H(X_2-X_1)-0.5H(X_2-X_1\mid U) \\
&\leq \log_2(3)+H(\overline{Y}_1\mid U)-0.5H(X_2+X_1\mid U)-0.5H(X_2-X_1\mid U)\\
&\leq \log_2(3)+H(\overline{Y}_1\mid U)-0.5H(X_2+X_1, X_2-X_1\mid U)\\
&= \log_2(3)+H(\overline{Y}_1\mid U)-0.5H(X_1,X_2\mid U)\\
&\leq \log_2(3)+H(\overline{Y}_1\mid U)-0.5H(\overline{Y}_1\mid U)\\
&\leq \log_2(3)+0.5H(\overline{Y}_1\mid U)\\
&\leq 1.5\log_2(3).
\end{align}
Step \eqref{eq:useGprop1} holds because $G\perp\!\!\perp U$. In Step \eqref{eq:useGprop2} we used the facts that $G\sim \mbox{Unif}(\{-1,1\})$, and $G\indep(X_1,X_2,U)$.
\hfill\qed

\subsection{Degraded Broadcast}
Next we consider a $2$-user degraded BC. In light of Theorem \ref{thm:same_marginal} it is not difficult to see that NS-assistance cannot improve the \emph{sum-capacity} of a degraded BC. This is because the same-marginals property (Theorem \ref{thm:same_marginal}) ensures that even with NS-assistance, a degraded BC is equivalent to a physically degraded BC, for which NS-assisted sum-capacity is upper bounded  by the classical capacity of the point to point channel obtained by allowing full cooperation between the two receivers (Lemma \ref{lem:p2p_capacity}). Classical capacity with full cooperation between the two receivers in a physically degraded BC is the same as the point to point channel capacity of the stronger receiver in the original degraded BC, which cannot exceed the classical sum-capacity of the original degraded BC. Thus, the NS-assisted sum-capacity of a degraded BC cannot exceed its classical sum-capacity. 

However, we will show in this section, by providing another toy example, that NS-assistance \emph{can} significantly improve the \emph{capacity region} of a degraded BC. The toy example is defined next.
\begin{definition}[Toy Channel $\mathcal{N}^{\mbox{\tiny toy2}}_{\overline{Y}_1,\overline{Y}_2\mid X_1,X_2}$]\label{def:toy2}
The toy channel $\mathcal{N}^{\mbox{\tiny \rm toy2}}_{\overline{Y}_1,\overline{Y}_2\mid X_1,X_2}$ is identical to $\mathcal{N}^{\mbox{\tiny \rm toy1}}_{\overline{Y}_1,\overline{Y}_2\mid X_1,X_2}$ (Definition \ref{def:toy1}) in every detail, but with one key difference. The output at Rx-$1$ is changed to $\overline{Y}_1=Y_1=(X_1,X_2)$. The resulting channel $\mathcal{N}^{\mbox{\tiny \rm toy2}}_{\overline{Y}_1,\overline{Y}_2\mid X_1,X_2}$ is physically degraded, because $(X_1,X_2)-\overline{Y}_1-\overline{Y}_2$ forms a Markov chain (and still semi-deterministic as well).
\end{definition}
For this  degraded BC, Theorem \ref{thm:degradedFq} establishes the improvement in capacity region due to NS-assistance, by comparing the highest rates achievable by Rx-$2$ with and without NS-assistance, given that Rx-$1$ achieves a rate $\log_2(q)$. Here also NS-assistance provides a factor of $2$ improvement. Moreover, this factor of $2$ improvement is established for all finite fields $\mathbb{F}_q$ except the degenerate\footnote{As mentioned previously (see the discussion preceding Theorem \ref{thm:semidetF3}), the $\mathbb{F}_2$ setting allows perfect CSIT (fixing $G=1$), which allows both receivers to simultaneously achieve the rate $\log_2(q)$, e.g., by setting $(X_1,X_2)=(W_1,W_2-W_1)$ for message symbols $W_1,W_2\in\mathbb{F}_q$.} case of $\mathbb{F}_2$. 
\begin{theorem}\label{thm:degradedFq}
For the degraded (and semi-deterministic) BC $\mathcal{N}^{\mbox{\tiny \rm toy2}}_{\overline{Y}_1,\overline{Y}_2\mid X_1,X_2}$ over any finite field $\mathbb{F}_q$, $q>2$,
\begin{align}
C^{\NS}_{2|R_1=\log_2(q)}&\triangleq \max\Big\{R_2: (\log_2(q),R_2)\in\mathcal{C}^{\NS}(\mathcal{N}^{\mbox{\tiny \rm toy2}})\Big\}=\log_2(q).\\
C_{2|R_1=\log_2(q)}&\triangleq\max\Big\{R_2: (\log_2(q),R_2)\in\mathcal{C}(\mathcal{N}^{\mbox{\tiny \rm toy2}})\Big\}=0.5\log_2(q).
\end{align}
\end{theorem}
\proof The lower bound $C^{\NS}_{2|R_1=\log_2(q)}\geq \log_2(q)$ follows because the NS-assisted coding scheme in Figure \ref{fig:NS_2user_bipartite} already achieves $(R_1,R_2)=(\log_2(q),\log_2(q))$, and the upper bound $C^{\NS}_{2|R_1=\log_2(q)}\leq \log_2(q)$ is immediate since $\log_2(q)$ is the single user capacity of Rx-$2$. For the classical case, we recall the capacity region of the 2-user degraded BC \cite[Thm. 5.2]{NIT}, applied to $\mathcal{N}^{\tiny \rm \mbox{toy 2}}$ as 
\begin{align}
	\mathcal{C}(\mathcal{N}^{\tiny \rm \mbox{toy 2}}) =\bigcup_{P_{U,X_1,X_2}}\left\{\begin{array}{ll}(R_1,R_2):\\
R_1\leq I(X_1,X_2; \overline{Y}_1\mid U)\\
R_2\leq I(U;\overline{Y}_2)\\
\end{array}\right\}.
\end{align}
It follows that any $(R_1,R_2)$ achievable by classical coding schemes satisfies
\begin{align}
	&0.5 R_1 + R_2 \notag \\
	&\leq 0.5I(X_1,X_2;\overline{Y}_1\mid U) + I(U;\overline{Y}_2) \\
	&= 0.5H(X_1,X_2\mid U) + I(U;GX_1+X_2\mid G) \label{eq:degraded_conv_use_indep} \\
	&= 0.5 H(X_1,X_2\mid U) + H(GX_1+X_2\mid G) - H(GX_1+X_2 \mid U,G) \\
	&\leq 0.5 H(X_1,X_2\mid U) +  \log_2(q) - \frac{1}{q-1}\sum_{g\in \mathbb{F}_q^{\times}} H(gX_1+X_2 \mid U) \label{eq:degraded_conv_use_indep_uniform} \\
	&= 0.5 H(X_1,X_2\mid U) +  \log_2(q) - \frac{1}{2(q-1)} \Big( \sum_{g\in \mathbb{F}_q^{\times}} H(g X_1+X_2 \mid U) + \sum_{g\in \mathbb{F}_q^{\times}} H(\pi(g)X_1+X_2 \mid U) \Big) \label{eq:degraded_conv_shift} \\
	&\leq 0.5 H(X_1,X_2 \mid U) +  \log_2 (q) -  \frac{1}{2(q-1)} \sum_{g\in \mathbb{F}_q^{\times}}H(gX_1+X_2, \pi(g)X_1+X_2\mid U) \\
	&= 0.5 H(X_1,X_2 \mid U) +  \log_2 (q) - \frac{1}{2(q-1)} \sum_{g\in \mathbb{F}_q^{\times}}H(X_1,X_2\mid U) \label{eq:degraded_conv_recover} \\
	&=0.5 H(X_1,X_2\mid U) +  \log_2 (q) - 0.5 H(X_1,X_2 \mid U) \\
	&= \log_2 (q)
\end{align}
and thus $C_{2\mid R_1=\log_2(q)}\leq 0.5\log_2(q)$. Step \eqref{eq:degraded_conv_use_indep} is because $G\perp\!\!\perp U$. In Step \eqref{eq:degraded_conv_use_indep_uniform} we use the facts that $G\sim {\rm Unif}(\mathbb{F}_q^{\times})$, and $G\perp\!\!\perp (X_1,X_2,U)$. In Step \eqref{eq:degraded_conv_shift}, $\pi$ is an invertible mapping (bijection) from $\mathbb{F}_q^{\times} \to \mathbb{F}_q^{\times}$ such that $\pi(g) \not = g, \forall g\in \mathbb{F}_q^{\times}$, e.g., $\pi(g) =\sigma g$ where $\sigma$ is the generator of $\mathbb{F}_q^{\times}$. Step \eqref{eq:degraded_conv_recover} then follows because there is a bijection between  $(X_1,X_2)$ and $(gX_1+X_2, \pi(g)X_1+X_2)$. Finally, since $(R_1,R_2) = (2\log_2(q),0)$ and $(R_1,R_2) = (0,\log_2 (q))$ are simply classically achievable by serving Rx-$1$ or Rx-$2$ individually, a time-sharing scheme achieves $(R_1,R_2) = (\log_2(q), 0.5\log_2(q))$. With this we conclude that $C_{2\mid R_1=\log_2(q)} =0.5\log_2(q)$.

\hfill\qed

\subsection{CoMP BC: Tree Networks}
Next we consider CoMP BC settings, starting with the tree network (Definition \ref{def:treenet}). See Fig. \ref{fig:CoMP} for an illustration of a generic tree network. Define the compact notation,
\begin{align}
d_\Sigma({\mbox{\footnotesize Tx-$b$}})\triangleq \sum_{k\in[K]: \mbox{\footnotesize Tx}(\mbox{\footnotesize Rx-$k$})=\mbox{\footnotesize Tx-$b$}}d_k,\\
R_\Sigma({\mbox{\footnotesize Tx-$b$}})\triangleq \sum_{k\in[K]: \mbox{\footnotesize Tx}(\mbox{\footnotesize Rx-$k$})=\mbox{\footnotesize Tx-$b$}}R_k.
\end{align}
In words, $d_\Sigma({\mbox{\footnotesize Tx-$b$}})$ is the sum of the DoF values of all receivers that are `associated'  with  Tx-$b$ (recall the definition of the `associated Tx' in \eqref{eq:associate}), and $R_\Sigma({\mbox{\footnotesize Tx-$b$}})$ is the sum of rates of all receivers that are `associated' with Tx-$b$.

\begin{theorem}[Tree Network DoF Region]\label{thm:tree_network}
	For a $K$-user CoMP BC tree network with tree graph $\mathcal{T} =$ $(\{\mbox{Tx-$0$}, \mbox{Tx-$1$},\cdots, \mbox{Tx-$B$}\}, \mathcal{E})$, under both the finite-field $\mathbb{F}_q$ (Section \ref{sec:Fqmodel}) and the real Gaussian (Section \ref{sec:gaussianmodel}) models, the classical (without NS-assistance) DoF region is characterized as,
	\begin{align} \label{eq:tree_region_classical}
		\mathcal{D}  = \left\{ (d_1,d_2,\cdots, d_K) \in \mathbb{R}_+^K \left|
		\begin{array}{c}
		\sum_{i=1}^L d_\Sigma(\mbox{\footnotesize Tx-$b_i$})\leq 1, \\
		\forall L \in \mathbb{N},~~ \forall \mbox{root-to-leaf paths of length $L$:}~(\mbox{\footnotesize Tx-${0}$},\mbox{\footnotesize Tx-${b_1}$}, \cdots, \mbox{\footnotesize Tx-${b_L}$})
		\end{array}
		\right.
		\right\}.
	\end{align}		
		For both the finite-field $\mathbb{F}_q$ (Section \ref{sec:Fqmodel}) and the real Gaussian (Section \ref{sec:gaussianmodel}) models, the NS-assisted DoF region is characterized as,
	\begin{align} \label{eq:tree_region_NS}
		\mathcal{D}^{\NS} = \left\{ (d_1,d_2,\cdots, d_K)  \in \mathbb{R}_+^K ~\left|	~~	d_\Sigma(\mbox{\footnotesize Tx-$b$})\leq 1,~~~\forall b\in[B]\right.
		\right\}.
	\end{align}
	The NS-assisted capacity region under the $\mathbb{F}_q$ model is characterized as,
		\begin{align} \label{eq:tree_cap_region_NS}
		\mathcal{C}^{\NS}(q) = \left\{ (R_1,R_2,\cdots, R_K)  \in \mathbb{R}_+^K ~\left|	~~	R_\Sigma(\mbox{\footnotesize Tx-$b$})\leq \log_2 q,~~~\forall b\in[B]\right.
		\right\}.
	\end{align}
\end{theorem}
The achievability and converse proofs of Theorem \ref{thm:tree_network} are provided for NS-assisted coding in Appendix \ref{proof:tree_NS}, and for classical coding in Appendix \ref{proof:tree_classical}. The following observations are in order.
\begin{enumerate}
\item The proof of achievability for classical coding schemes in Section \ref{sec:tdma} is based on a simple \emph{weighted tree graph-burning} argument, and shows that an orthogonal scheduling strategy, namely time-division multiple access (TDMA), is optimal for all tree networks. Connections to graph-burning literature \cite{Bonato} for classical coding schemes in CoMP BC settings beyond tree-networks may be worth exploring.
\item For the proof of converse for classical coding schemes, presented in Section \ref{sec:AIS}, statistical equivalence, i.e., same-marginals arguments turn out to be  too weak/loose. Consider for example a $K\times K$ connectivity matrix ${\bf M}$ that has only zeros above the main diagonal, and only *'s on the main diagonal and below. For such a CoMP BC, same-marginals arguments only produce trivial bounds, e.g., that the sum-DoF cannot be more than $K$, when in fact the sum-DoF cannot be more than $1$. There are other bounds, e.g., based on compound BC arguments as in \cite{Weingarten_Shamai_Kramer, Gou_Jafar_Wang}, but those bounds are also not tight enough in general to establish a sum-capacity advantage greater than a factor of $2$ due to NS-assistance, when in fact this advantage in a $K$ user BC can be as large as a factor of $K$. 
To show the tighter converse (which allows us to establish the factor of $K$ gain from NS), we adapt the Aligned Images approach\footnote{The AIS bound \cite{davoodi2016aligned} is the only existing bound to our knowledge that is capable of establishing the factor of $K$ advantage in such a setting.} of \cite{davoodi2016aligned} to the $\mathbb{F}_q$ model. This adaptation may be of independent interest, as the finite field model provides a potentially more tractable setting to explore further strengthening of the AIS bounds.
\item A proof of achievability for NS-assisted coding schemes is first presented in Section \ref{sec:oneproof} based on a $K+1$ partite NS box. An alternative achievability proof is also presented  in Section \ref{sec:altproof}, based on a \emph{successive encoding} strategy that employs only bipartite NS boxes of the OTP (one-time pad) type identified in \cite{OTPmodel}. The successive encoding strategy is essentially a `wiring' \cite{allcock2009closed, beigi2015monotone} of bipartite boxes. The insights from this bipartite NS box based approach are extended to other settings, e.g., fading dirty paper channel in Section \ref{sec:fadingdirt}. The  $K+1$ partite NS box construction in Section \ref{sec:oneproof} is included because it has the potential to generalize further, to (non tree-network) CoMP settings where the capacity remains open, such as \eqref{eq:fanomatrix}. The necessity of $N$ partite NS boxes with $N>2$ to achieve capacity for general CoMP settings remains an interesting open question.
\end{enumerate}

Most importantly,  in sharp contrast with the fully-connected CoMP BC which does not benefit from NS-assistance, Theorem \ref{thm:tree_network} reveals surprisingly significant advantages of NS-assistance in tree networks. To explore this aspect further, consider the sum-DoF metric. From Theorem \ref{thm:tree_network}, direct characterizations of the sum-DoF values with and without NS-assistance are obtained immediately as the following corollary.

\begin{corollary}[Tree Network Sum-DoF] \label{cor:sumDoF}
For a $K$-user CoMP BC tree network with tree graph $\mathcal{T}$, under both the finite-field and  real Gaussian models, the sum-DoF with and without NS-assistance, respectively, are characterized as, 
\begin{align}
	d_{\Sigma} &= \ell(\mathcal{T})=\mbox{Number of leaf nodes in $\mathcal{T}$},\\
	d_{\Sigma}^{\NS} &= B = \mbox{Number of (non-root) nodes in $\mathcal{T}$}.
\end{align}
\end{corollary}
\begin{proof}
	According to Theorem \ref{thm:tree_network}, for NS-assisted coding schemes, $d_{\Sigma}^{\NS} \geq B$, because $d_\Sigma(\mbox{Tx-$1$})=d_\Sigma(\mbox{Tx-$2$})=\cdots=d_\Sigma(\mbox{Tx-$B$})=1$ satisfies the constraint \eqref{eq:tree_region_NS} and each Rx is associated with at most one Tx. On the other hand, $d_{\Sigma}^{\NS} \leq B$ because every Rx is associated with at least one Tx, and all Txs are accounted for. For classical coding schemes, $d_{\Sigma} \geq \ell(\mathcal{T})$ because to satisfy \eqref{eq:tree_region_classical} one can set $d_\Sigma(\mbox{Tx-$i$}) = 1$ for all $i$ such that \mbox{Tx-$i$} is a leaf node of $\mathcal{T}$, and $d_\Sigma(\mbox{Tx-$j$}) = 0$ for all $j$ such that \mbox{Tx-$j$} is not a leaf node of $\mathcal{T}$. On the other hand, $d_{\Sigma}\leq \ell(\mathcal{T})$ as there are $\ell(\mathcal{T})$ bounds in the RHS of \eqref{eq:tree_region_classical}, and adding these bounds (along with the non-negativity of DoF) implies that $d_{\Sigma} \leq \ell(\mathcal{T})$. 
\end{proof}
 
Corollary \ref{cor:sumDoF} shows that in terms of sum-DoF of tree networks, the benefits of NS-assistance are most significant for those tree graphs where most nodes are not leaf-nodes. The extreme case thus becomes apparent as a path graph, which has $K$ (non-root) nodes, but only one leaf node, e.g., the setting in Figure \ref{fig:bestcaseexample}. The transmit antennas form a tree graph shown in the left-hand side of Fig. \ref{fig:tree_structures}. On the other hand, there is no DoF benefit of NS-assistance for a star graph (the right-hand side of Fig. \ref{fig:tree_structures}), for which all non-root nodes are leaf nodes.

\begin{figure}[h]
\begin{tikzpicture}
\node (myfirstpic) at (0,0) {\includegraphics[width=0.65\textwidth]{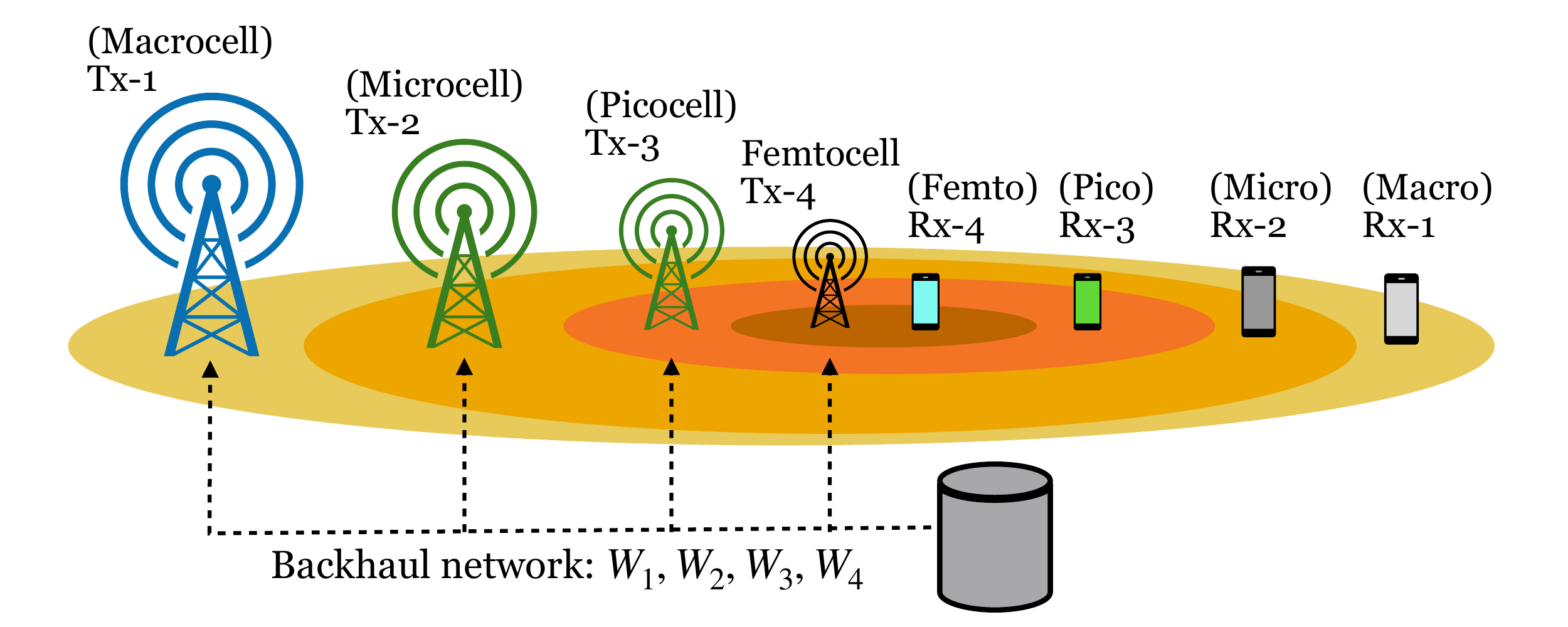}};
\node at (8.1,0){\(
   \kbordermatrix{
    & \mbox{\tiny Tx-$1$} & \mbox{\tiny Tx-$2$} & \mbox{\tiny Tx-$3$} & \mbox{\tiny Tx-$4$} \\
    \mbox{\tiny Rx-$1$} & * & 0 & 0 & 0  \\
    \mbox{\tiny Rx-$2$} & * & * & 0 & 0  \\
    \mbox{\tiny Rx-$3$} & * & * & * & 0  \\
    \mbox{\tiny Rx-$4$} & * & * & * & *  \\
  }={\bf M}
  \)};
  \end{tikzpicture}
\caption{A CoMP BC with a $K=4$ level vertical hierarchical structure where NS-assistance improves DoF by a factor of $K=4$ compared to classical DoF. }
\label{fig:bestcaseexample}
\end{figure}

\begin{figure}[h]
\center
\begin{tikzpicture}
\draw[very thin, rounded corners=2pt] (-1.5,-1.5) rectangle (5.5,0.5);
\node (Tx0) at (0,0) {{\footnotesize $(\mbox{Root Tx-}0)$}};
\node (Tx1) at (0,-1) {{\footnotesize $\mbox{Tx-}1$}};
\node (Tx2) at (1.5,-1) {{\footnotesize $\mbox{Tx-}2$}};
\node (Tx3) at (3,-1) {{\footnotesize $\mbox{Tx-}3$}};
\node (Tx4) at (4.5,-1) {{\footnotesize $\mbox{Tx-}4$}};

\draw[thick] (Tx0) -- (Tx1);
\draw[thick] (Tx1) -- (Tx2);
\draw[thick] (Tx2) -- (Tx3);
\draw[thick] (Tx3) -- (Tx4);

\node[rectangle, draw] at (4,0) {\footnotesize $d_{\Sigma}^{\NS}/d_{\Sigma} = 4$};
 
\begin{scope}[shift= {(8,0)}]
\draw[very thin, rounded corners=2pt] (-1.5,-1.5) rectangle (5.5,0.5);
\node (Tx0) at (0,0) {{\footnotesize $(\mbox{Root Tx-}0)$}};
\node (Tx1) at (0,-1) {{\footnotesize $\mbox{Tx-}1$}};
\node (Tx2) at (1.5,-1) {{\footnotesize $\mbox{Tx-}2$}};
\node (Tx3) at (3,-1) {{\footnotesize $\mbox{Tx-}3$}};
\node (Tx4) at (4.5,-1) {{\footnotesize $\mbox{Tx-}4$}};

\draw[thick] (Tx0) -- (Tx1);
\draw[thick] (Tx0) -- (Tx2);
\draw[thick] (Tx0) -- (Tx3);
\draw[thick] (Tx0) -- (Tx4);

\node[rectangle, draw] at (4,0) {\footnotesize $d_{\Sigma}^{\NS}/d_{\Sigma} = 1$};
\end{scope}
\end{tikzpicture}
\caption{Two extremal tree structures of transmit antennas: The left-hand side is a path graph, which has $4$ (in general $K$) non-root nodes but only one leaf node, being Tx-$4$ (in general Tx-$K$). The right-hand side is a star graph, which contains as many leaf nodes as non-root nodes.} 
\label{fig:tree_structures}
\end{figure}

\begin{theorem}[Extremal gain from NS-assistance in a $K$ user BC] \label{thm:bestcase}
Let $\mathfrak{N}_K$ be the set of all $K$-user broadcast channels, and $C_\Sigma^{\NS}(\mathcal{N}), C_\Sigma(\mathcal{N})$ the sum-capacity of $\mathcal{N}\in\mathfrak{N}_K$ with and without NS-assistance, respectively. Then,
\begin{align}
\sup_{\mathcal{N}\in\mathfrak{N}_K}\frac{C_\Sigma^{\NS}(\mathcal{N})}{C_{\Sigma}(\mathcal{N})}=K.
\end{align}
\end{theorem}
\begin{proof}
Asymptotic ($q\rightarrow\infty$) achievability of a factor of $K$ gain in Shannon capacity due to NS-assistance in a CoMP BC corresponding to a path graph  is implied directly by Corollary \ref{cor:sumDoF}. That a gain by a factor larger than $K$ is impossible in \emph{any} $K$-user BC follows from the observation that in the NS-assisted coding scheme, by serving only the one user that has the highest single-user capacity of all $K$ users, the rate achieved is at least $C^{\NS}_\Sigma(\mathcal{N})/K$, but since this is a rate achievable with a single receiver (i.e., a point to point channel), it is also achievable without NS-assistance (Lemma \ref{lem:p2p_capacity}), i.e., $C_\Sigma(\mathcal{N})\geq C^{\NS}_\Sigma(\mathcal{N})/K$, which completes the proof. 
\end{proof}

\subsection{CoMP BC: General Connectivity}
We focus only on the $\mathbb{F}_q$ model, sum-capacity and sum-DoF with NS-assistance in this section. We need the following definitions.

\begin{definition}[Min-rank]\label{def:minrk}
Given a channel connectivity matrix ${\bf M}\in \{0,\ast\}^{K\times B}$, define $$\minrk({\bf M}) \triangleq \min_{\overline{\bf G} \in \mathcal{G}({\bf M})}\rank(\overline{\bf G}),$$ where $\mathcal{G}({\bf M}) \triangleq \{\overline{\bf G} \in \mathbb{F}_q^{K\times B}\colon \overline{\bf G} ~\mbox{fits}~ {\bf M}\}$. We say that $\overline{\bf G}$ fits ${\bf M}$ if $[M_{ij}=0]\iff [\overline{G}_{ij}=0]$ for all $(i,j)\in[K]\times[B]$.
\end{definition}

\begin{theorem}[Min-rank converse] \label{thm:Sato_min_rank}
	Given a CoMP BC over $\mathbb{F}_q$ (Section \ref{sec:Fqmodel}) with channel connectivity matrix ${\bf M}\in\{0,*\}^{K\times B}$, the NS-assisted sum-capacity (for any finite field $\mathbb{F}_q$) and sum-DoF values are bounded from above as, 
	\begin{align} 
	\max\left(\frac{C_{\Sigma}^{\NS}(q)}{\log_2q},~d_{\Sigma}^{\NS}\right)  &\leq \minrk({\bf M}).
	\end{align}
\end{theorem}
\begin{proof}
Consider any $\overline{\bf G} \in \mathcal{G}({\bf M})$, and let $\lambda_1,\lambda_2,\cdots,\lambda_B$ be $B$ random variables chosen independently uniformly from $\mathbb{F}_q^\times$. Then 
\begin{align}
	\widetilde{\bf G} \triangleq \overline{\bf G} \times \diag([\lambda_1,\lambda_2,\cdots, \lambda_B])
\end{align}
represents a non-zero scaling of the $b^{th}$ column of $\overline{\bf G}$ with $\lambda_b$, for $b\in [B]$.  Denote by $\widetilde{\bf G}_k$ the $k^{th}$ row of $\widetilde{\bf G}$. Then $\widetilde{\bf G}_k$ has the same marginal distribution as ${\bf G}_k$, as defined in \eqref{eq:Ykout}, for all $k\in [K]$. By Theorem \ref{thm:same_marginal}, the NS-assisted capacity of the CoMP BC with channels ${\bf G}$ is the same as the NS-assisted capacity of the CoMP BC with channels $\widetilde{\bf G}$. Now for the CoMP BC with channels $\widetilde{\bf G}$, let all the receivers collaborate, resulting in a point-to-point communication problem over a MIMO channel with $B$ transmit antennas and $K$ receive antennas, for which NS-assisted capacity is equal to the classical capacity (Lemma \ref{lem:p2p_capacity}). It is known that the for a MIMO channel, the sum-capacity value is equal to the rank of the channel matrix in $q$-ary units, i.e., $\log_2q$ times the rank of the channel matrix in binary units.  Thus, $\rank(\widetilde{\bf G})\log_2q$  serves as an upper bound on the NS-assisted sum-capacity of the CoMP BC, because receiver collaboration cannot make the sum-capacity smaller. Finally, since non-zero scaling of columns does not change the rank, $\rank(\widetilde{\bf G}) = \rank(\overline{\bf G})$. The bound on sum-DoF follows from a normalization of sum-capacity by $\log_2q$.
\end{proof}

\begin{definition}[$D$-triangular matrix]
	Given a channel connectivity matrix ${\bf M}$, we say that ${\bf M}$ contains a $D$-triangular matrix if there exist permutations of the rows and columns of ${\bf M}$ that yield as a submatrix of ${\bf M}$, a $D\times D$ lower triangular matrix with only $\ast$’s on the main diagonal.
\end{definition}
 
\begin{definition}[Triangle number]\label{def:tnumber}
	Given a channel connectivity matrix ${\bf M}$, let the triangle number be defined as $$\tri({\bf M}) \triangleq  \max\{D\in \mathbb{N} \colon {\bf M}  ~\mbox{contains a}~D\mbox{-triangular matrix}\}.$$
\end{definition}

\begin{theorem}[Triangle achievability] \label{thm:triangle_achi}
	Given a CoMP BC over $\mathbb{F}_q$ with channel connectivity matrix ${\bf M}$, the NS-assisted sum-capacity (for any finite field $\mathbb{F}_q$) and sum-DoF values  are bounded from below as, 
	\begin{align}
	\min\left(\frac{C_{\Sigma}^{\NS}(q)}{\log_2q},~d_{\Sigma}^{\NS}\right) & \geq \tri({\bf M}).
	\end{align}
\end{theorem}
Say $\tri({\bf M}) = D$. By the definition of triangle number, there exists a submatrix of ${\bf M}$ which is a $D\times D$ lower triangular matrix with only $\ast$'s on the main diagonal. Now suppose  only those transmit antennas are active that correspond to the columns of the submatrix  and that only those receivers are served that correspond to the rows of the submatrix. Then it remains to show the achievability of $d_{\Sigma}^{\NS} \geq D$ for the sub-network. This is proved by Remark \ref{rem:triangle_number} of Section \ref{proof:tree_NS}.

\begin{lemma}[Lemma 2.1 of \cite{johnson2008extent}] \label{lem:r_nonzeros}
	For ${\bf M}\in \{0,\ast\}^{K\times B}$, if each column (or each row) of ${\bf M}$ contains at least $r$ occurrences of `$\ast$,'    then $\minrk({\bf M})\leq K+1-r$.
\end{lemma}

\begin{corollary}
	For $r\in [K]$, if  every transmit antenna is connected to at least $r$ receivers, then for any finite field $\mathbb{F}_q$, $$\max\left(\frac{C_{\Sigma}^{\NS}(q)}{\log_2q},~d_{\Sigma}^{\NS}\right)\leq K+1-r.$$ 
	Similarly, if every receiver is connected to at least $r$ transmit antennas, then for any finite field $\mathbb{F}_q$,$$\max\left(\frac{C_{\Sigma}^{\NS}(q)}{\log_2q},~d_{\Sigma}^{\NS}\right)\leq K+1-r.$$
\end{corollary}
\begin{proof}
	The corollary is directly implied by Theorem \ref{thm:Sato_min_rank} and Lemma \ref{lem:r_nonzeros}. 
\end{proof}

\begin{lemma}[Prop. 4.6 of \cite{minrk}] \label{lem:PropOfMinrk}
	If $\minrk({\bf M}) = \min\{K,B\}$, then $\tri({\bf M})= \minrk({\bf M}) =  \min\{K,B\}$.
\end{lemma}

\begin{corollary}
	For the $K$-user MISO BC with channel connectivity matrix ${\bf M}$, if $\frac{C_{\Sigma}^{\NS}(q)}{\log_2q}= \min\{K,B\}$ or $d_{\Sigma}^{\NS}({\bf M})= \min\{K,B\}$, then $\tri({\bf M}) = \min\{K,B\}$.
\end{corollary}
\begin{proof}
	By Theorem \ref{thm:Sato_min_rank} and the fact that $\minrk({\bf M})\leq \min\{K,B\}$, it follows that $\minrk({\bf M}) = \min\{K,B\}$. The result then follows from Lemma \ref{lem:PropOfMinrk}.
\end{proof}
\noindent In other words, if  the NS-assisted DoF is the largest possible, which is $\min\{K,B\}$, then ${\bf M}$ must contain a $\min\{K,B\}$-triangular matrix.

\begin{lemma}[Table in Sec. 6 of \cite{johnson2008extent}] \label{lem:small_T}
	For ${\bf M} \in \{0,\ast\}^{K\times B}$,  $\minrk({\bf M}) = \tri({\bf M})$ if $\min\{K,B\}\leq 6$.
\end{lemma}

\begin{corollary}\label{cor:upto6}
	Given a $K$-user CoMP BC over $\mathbb{F}_q$ with $B$ transmit antennas and channel connectivity matrix ${\bf M}\in\{0,*\}^{K\times B}$,  if $\min\{K,B\}\leq 6$, then for any finite field $\mathbb{F}_q$,
\begin{align}
\frac{C_{\Sigma}^{\NS}(q)}{\log_2q} = d_{\Sigma}^{\NS}= \minrk({\bf M}) = \tri({\bf M}).
\end{align}
\end{corollary}
\begin{proof}
	The corollary is implied by Theorem \ref{thm:Sato_min_rank}, Theorem \ref{thm:triangle_achi} and Lemma \ref{lem:small_T}.
\end{proof}
In light of Corollary \ref{cor:upto6}, a $K=B=7$ is the smallest CoMP setting where the sum-capacity/DoF with NS-assistance remains open. The setting is challenging partly because it includes the connectivity matrix corresponding to the Fano projective plane \cite[Example 4.3]{minrk} as shown below.
\begin{align}\label{eq:fanomatrix}
{\bf M}&=\begin{bmatrix}
*&0&0&0&*&*&*\\
0&*&0&*&0&*&*\\
0&0&*&*&*&0&*\\
0&*&*&0&*&*&0\\
*&0&*&*&0&*&0\\
*&*&0&*&*&0&0\\
*&*&*&0&0&0&*
\end{bmatrix}.
\end{align}
For this ${\bf M}$, the triangle-number is strictly smaller than the minrank over any $\mathbb{F}_q$ where $q$ is not a power of $2$. Specifically, the triangle number is field-independent, and is equal to $3$ in this case, but the minrank is field-dependent, equal to $3$ if $q$ is a power of $2$, and equal to $4$ otherwise. Thus, we have the sum-capacity $C_\Sigma(q)=3\log_2q$ if $q$ is a power of $2$, while a gap remains, i.e.,  $3\log_2q\leq C_\Sigma(q)\leq 4\log_2q$ if $q$ is not a power of $2$. Closing this gap is an interesting open problem. On the other hand, the sum-DoF for this connectivity matrix can be shown to be $3$ as follows. The lower bound on sum-DoF is already implied by the triangle-number. The upper bound is implied by the following observation. In order to achieve a sum-DoF value greater than $3$, there must be a sequence of rates $R_{\Sigma}(q)$, indexed by the field size $q$, such that $\lim_{q\to \infty} R_{\Sigma}(q)/\log_2 q > 3$. However, no matter how large $q$ is, there is always a power of $2$ that is larger than $q$. Since we also know that $C_{\Sigma}(q)\leq 3\log_2 q$ whenever $q$ is a power of $2$, this means that there can be no such sequence of $R_{\Sigma}(q)$ satisfying $\lim_{q\to \infty} R_{\Sigma}(q)/\log_2 q > 3$.

\section{Extensions: Communication with Side-information}
Considering that any potential for a $K$-fold capacity/DoF improvement can be quite significant in wireless networks, it is worthwhile to search for other settings where similar gains may be found.  In this section, we identify the problem of \emph{communication with side-information} as such a setting. The definitions of classical and NS-assisted coding schemes, probabilities of error, achievable rates, capacity and DoF regions are adapted to the settings considered in this section in a straightforward manner, we will omit these repetitive details. Also, while noting that the results translate to Gaussian settings as usual, let us consider only $\mathbb{F}_q$ models in this section for simplicity. 

\subsection{Fading Dirty Paper Channel}\label{sec:fadingdirt}
Consider a point-to-point channel where, over the $\tau^{th}$ channel use,  the output at the Rx is,
\begin{align} \label{eq:def_fading_dirt}
	\overline{Y}^{(\tau)}&=(Y^{(\tau)},G^{(\tau)})\\
	Y^{(\tau)} &= X^{(\tau)} + G^{(\tau)} \Theta^{(\tau)},
\end{align}
and the input from the Tx is $X^{(\tau)}$. All symbols and operations are in $\mathbb{F}_q$. The received signal includes additive interference $\Theta^{(\tau)}$, scaled by a random channel fading coefficient $G^{(\tau)}$. Say $\Theta^{(\tau)}$ and $G^{(\tau)}$ are independent and uniformly distributed over $\mathbb{F}_q$. It is assumed that $\Theta^{(\tau)}$ is known in advance to the transmitter (non-causal side-information) but not to the receiver, while $G^{(\tau)}$ is known to the receiver but not to the transmitter. This corresponds to what is known in the literature as the \emph{fading dirty paper channel} \cite{Rini_Shamai_fading_dirt}. The following theorem shows that the multiplicative gain from NS-assistance is unbounded for such a channel.

\begin{theorem}\label{thm:fadingdirt}
For the  fading dirty paper channel defined in this section, the NS-assisted capacity is $C^{\NS}=\log_2(q)$, while the classical capacity is $o_q(\log_2 q)$, i.e., vanishingly small relative to $\log_2 q$ as $q\rightarrow\infty$. Therefore, the multiplicative capacity gain from NS-assistance is unbounded as $q\to \infty$.
\end{theorem}
\begin{proof}
	To show $C^{\NS} \leq \log_2 q$, suppose $\Theta^{[n]}$ can be designed by the transmitter, as this cannot reduce the NS-assisted capacity and thus provides a valid upper bound for it. Now the channel becomes a point-to-point channel with input $(X^{(\tau)}, \Theta^{(\tau)})$ and output $(Y^{(\tau)}, G^{(\tau)})$. From Lemma \ref{lem:p2p_capacity}, the NS-assisted capacity for the point-to-point channel is the same as the classical capacity, which is $$\max_{P_{X^{[n]}\Theta^{[n]}}}I(X^{[n]}, \Theta^{[n]}; Y^{[n]}, G^{[n]}) = \max_{P_{X^{[n]}\Theta^{[n]}}}I(X^{[n]}, \Theta^{[n]}; Y^{[n]} \mid G^{[n]}) \leq \log_2 q,$$ as $G^{[n]}$ is independent of $(X^{[n]},G^{[n]})$ and $H(Y^{[n]})\leq \log_2 q$.
	
	 To show $C^{\NS} \geq \log_2 q$, let the Tx and the Rx share an NS box $\mathcal{Z}$. The inputs for the Tx, Rx are $U\in \mathbb{F}_q,V\in \mathbb{F}_q$ and the outputs for the Tx, Rx are denoted as $S\in \mathbb{F}_q,T\in \mathbb{F}_q$, respectively. The input-output relationship is specified as
	 \begin{align}
\begin{array}{|c|c|c|c|}
		\hline
\mbox{Party}&  \mbox{Tx}  &\mbox{Rx}\\\hline
\mbox{Input}&S  &T\\\hline
\mbox{Output}&U &V=U + ST\\\hline
\end{array}
\end{align}
where $U$ is a random variable uniformly distributed over $\mathbb{F}_q$.  All operations are defined in $\mathbb{F}_q$. The box being NS is verified as any one of the parties can learn nothing about the input of the other party.\footnote{Such a box belongs to the class of NS boxes referred to as the OTP (one-time pad) model in \cite{OTPmodel}. } The coding scheme will use the channel only once, which allows us to omit the channel use index. The message set is  $\mathcal{M} = \mathbb{F}_q$. The Tx inputs $S=\Theta$ to the NS box, and obtains $U$ as its output from the box. The input to the channel is chosen as $X=W+U$. At the Rx, the input to the NS box is $T = G$ (the channel coefficient). From the channel the  Rx obtains $Y = X+G\Theta = W+U + G\Theta = W+V$. The Rx then subtracts $V$ from $Y$ to obtain $W$. The scheme shows that $R = \log_2 q$ is achievable by NS coding schemes. Thus, $C^{\NS} = \log_2 q$.
	 
For classical coding, we only need to show that $C \leq o_q(\log_2 q)$. This requires the use of the AIS bound. Suppose a sequence (indexed by the number of channel uses utilized by the scheme, $n$) of coding schemes (with the $n^{th}$ scheme having message set $\mathcal{M}^{(n)}$) achieves rate $R$. Fano's inequality implies,
\begin{align}
 	&\log_2|\mathcal{M}^{(n)}| - o(n) \notag \\
 	&\leq I(W;Y^{[n]} \mid G^{[n]}) \\
 	&= H(Y^{[n]}\mid  G^{[n]}) - H(Y^{[n]}\mid  G^{[n]},W) \\
 	&\leq n\log_2 q  - H(Y^{[n]} \mid G^{[n]}, W)  \\
 	&=H(\Theta^{[n]} \mid G^{[n]}, W) - H(Y^{[n]} \mid G^{[n]}, W) \label{eq:fading_dirt_theta} ~~~~ \mbox{($\because \Theta^{[n]}$ is independent of $(G^{[n]}, W)$)}\\
 	&\leq H(\Theta^{[n]} \mid G^{[n]}, W=w^*) - H(Y^{[n]} \mid G^{[n]}, W=w^*), ~ ~~~~ \mbox{(there exists such a $w^*$)}\\
 	&\leq \max_{P_{X^{[n]}}} \big( H(\Theta^{[n]} \mid G^{[n]}) - H(Y^{[n]} \mid G^{[n]}) \big) \\
 	&\leq n o_q(\log_2 q) ~~ \mbox{($\because$ AIS bound)} \\\
 	& \hspace{-1cm} \implies R \leq \lim_{n\to \infty}\log_2|\mathcal{M}^{(n)}|/n \leq o_q(\log_2 q)
\end{align}
This completes the proof.
\end{proof}

\subsection{$K$-user MAC with Side-information}
We generalize the point to point fading dirty paper channel into a multiple-access setting, with the additional interesting aspect that our achievability in this case is based on an NS box that does not {seem} to be constructible from bipartite NS boxes. Consider the $K$-user discrete-memoryless MAC with input-output relationship,
\begin{align} \label{eq:def_DMMAC}
	Y^{(\tau)} = X_1^{(\tau)} + X_2^{(\tau)} + \cdots + X_K^{(\tau)} + f\big(\Theta_0^{(\tau)}, \Theta_1^{(\tau)},\cdots, \Theta_K^{(\tau)}\big) + Z^{(\tau)},
\end{align}
for the $\tau^{th}$ channel use. Here, $(\Theta_0^{(\tau)},\Theta_1^{(\tau)},\cdots, \Theta_K^{(\tau)})$ comprise the channel state, with $\Theta_0^{(\tau)}$  only available at the Rx, and $\Theta_k^{(\tau)}$  only available  at Tx-$k$ (as non-causal side-information), for $k\in [K]$. The additive term $f\big(\Theta_0^{(\tau)}, \Theta_1^{(\tau)},\cdots, \Theta_K^{(\tau)}\big)$ is referred to as the interference, and it depends on the channel state. We do not pose constraints on the domain of $f$ but we require that the codomain of $f$ should be $\mathbb{F}_q$. $X_k^{(\tau)}\in \mathbb{F}_q$ is the input to the channel at Tx-$k$ for $k\in [K]$. $Z^{(\tau)}\in \mathbb{F}_q$ is  additive noise, assumed independent of the inputs and the channel states. $Y^{(\tau)}$ is the output seen by the Rx. There are $K$ independent messages $W_1,W_2,\cdots, W_K$ such that $W_k$ originates at Tx-$k$, $k\in [K]$, and the Rx needs to decode all $K$ messages. 

To clarify the notation let us specify that an NS-assisted coding scheme (over $n$ channel uses) utilizes a $K+1$ partite NS box $\mathcal{Z}$, with $S_k$ denoting the input at Tx-$k$, $T$ denoting the input at the Rx, $U_k$ denoting the output at Tx-$k$, and $V$ denoting the output at the Rx. For $k\in [K]$, the input of the NS box at Tx-$k$ is set as $S_k = (W_k, \Theta_k^{[n]})$, and the input to the channel is set as $X_k^{[n]} = U_k$. At the Rx, the input of the box is set as $T = (Y^{[n]}, \Theta_0^{[n]})$, and the output of the box are the decoded messages $(\widehat{W}_1,\cdots, \widehat{W}_K) = V$.

\begin{theorem} \label{thm:DMMAC}
	For the channel defined in \eqref{eq:def_DMMAC}, the NS capacity region $\mathcal{C}^{\NS}$ contains the classical (without NS-assistance) capacity region of the MAC where the interference term is absent, i.e.,
	\begin{align} \label{eq:def_DMMAC_no_interference}
		Y^{(\tau)} = X_1^{(\tau)} + X_2^{(\tau)} + \cdots + X_K^{(\tau)}   + Z^{(\tau)}.
	\end{align}
\end{theorem}
\begin{proof}
We show how to convert the channel \eqref{eq:def_DMMAC} into the channel \eqref{eq:def_DMMAC_no_interference} using a NS box. Define a $K+1$ partite NS box as,
\begin{align}
\begin{array}{|c|c|c|c|}
		\hline
\mbox{Party}&  \mbox{Tx-$k$}  &\mbox{Rx}\\\hline
\mbox{Input}&S_k &T\\\hline
\mbox{Output}&U_k&V=U_1+\cdots +U_K+f(T,S_1,\cdots, S_K)\\\hline
\end{array}
\end{align}
where $(U_1,U_2,\cdots, U_K)$ is uniformly distributed over $\mathbb{F}_q^K$. It is not difficult to verify that the box is NS, as any $K$ parties collaborating together can learn nothing about the input of the remaining party. Omitting the channel use index, let $S_k = \Theta_k$ for $k\in [K]$, $T=\Theta_0$. Meanwhile, let the input to the channel be $X_k = U_k+\overline{X}_k$ for $k\in [K]$, where $\overline{X}\in \mathbb{F}_q$, the output of the channel is then
\begin{align}
	Y = \overline{X}_1+\cdots + \overline{X}_K + U_1+\cdots +U_K  + f(\Theta_0,\cdots, \Theta_K) + Z.
\end{align}
The Rx subtracts $V$ from $Y$ to obtain $\overline{Y} = Y-V = \overline{X}_1+\cdots + \overline{X}_K +Z$. The resulting channel with inputs $\overline{X}_1,\cdots, \overline{X}_K$ and output $\overline{Y}$ has the form in \eqref{eq:def_DMMAC_no_interference}. Therefore, any classical coding scheme for the channel without interference can be applied in the converted channel to achieve the same rate tuple.
\end{proof}

\section{Conclusion}\label{sec:conclusion}
The discovery of a $K$-fold increase in the high-SNR Shannon capacity (DoF) of a wireless network due to NS-assistance, leads to many follow up questions, such as --  what other wireless network settings can benefit  significantly from NS-assistance? is channel uncertainty a critical requirement for such settings? and how much of this capacity improvement is achievable with quantum resources?  In particular, the $K$-fold improvement via NS-assistance leaves  hope that even if quantum resources are able to recover a \emph{fraction} of this improvement, that could still be quite significant. Unlike NS-assisted capacity which can be formulated and studied through purely classical information theoretic tools, finding the quantum-assisted capacity would require quantum-theoretic analysis, which is left to future works.

\appendix
\section{Proof of Theorem \ref{thm:same_marginal}}\label{sec:proofsamemarginal}
We have two $K$-user broadcast channels $\mathcal{N}_{Y_1\cdots Y_K\mid X}$ and $\widetilde{\mathcal{N}}_{Y_1\cdots Y_K\mid X}$, such that the marginal distributions $\mathcal{N}_{Y_k\mid X} = \widetilde{\mathcal{N}}_{Y_k\mid X}$ for all $k\in [K]$. Suppose we are given $\mathcal{Z}$ as the NS box for a NS coding scheme over $n$ channel uses. 
Recall that for $k\in [K]$, $P_{e,k}$ is the error probability for the $k^{th}$ message. Say $P_{e,k} = \epsilon_k$ if $\mathcal{Z}$ is applied to channel $\mathcal{N}_{Y_1\cdots Y_K\mid X}$, and  $P_{e,k} = \widetilde{\epsilon}_k$ if $\mathcal{Z}$ is applied to channel $\widetilde{\mathcal{N}}_{Y_1\cdots Y_K\mid X}$. In the following we prove that $\epsilon_k = \widetilde{\epsilon}_k, \forall k\in [K]$, thus showing the two BCs have the same NS-assisted capacity region. Without loss of generality, we prove this for $k=1$.

Since $\mathcal{Z}$ is NS, we can define,
\begin{align}
	&\sum_{\hat{w}_2,\cdots, \hat{w}_K} \mathcal{Z}(x^{[n]}, w_{1},\hat{w}_2,\cdots,\hat{w}_{K}  \mid  [w_1,w_2,\cdots,w_K], y_{1}^{[n]},y_{2}^{[n]},\cdots,y_{K}^{[n]}) \notag \\
	& ~~~~~~~~~~\triangleq
	\mathcal{Z}_{01}(x^{[n]}, w_{1}   \mid  [w_1,w_2,\cdots,w_K], y_{1}^{[n]} ),
\end{align}
as the result does not depend on $y_2^{[n]},\cdots, y_K^{[n]}$.
Then according to \eqref{eq:NS_joint_prob}, 
{\small
\begin{align}
	&1-\epsilon_1 = \sum_{w_1\in \mathcal{M}_1} \Pr(W_1=w_1,\hat{W}_1=w_1) \\
	&= \sum_{w_{[K]},\hat{w}_{[K]\setminus\{1\}},x^{[n]}, y_{[K]}^{[n]}}\Pr(W_1 = \widehat{W}_1=w_1, W_{[K]\setminus \{1\}}=w_{[K]\setminus \{1\}}, \widehat{W}_{[K]\setminus \{1\}}=\hat{w}_{[K]\setminus \{1\}}, X^{[n]}= x^{[n]}, Y_{[K]}^{[n]} = y_{[K]}^{[n]} )   \\
	&= \frac{1}{\prod_{k=1}^K|\mathcal{M}_k|} \sum_{w_{[K]},x^{[n]},y^{[n]}_{1}} \sum_{y^{[n]}_{2}, \cdots, y^{[n]}_{K}} \sum_{\hat{w}_2,\cdots, \hat{w}_K} \mathcal{Z}(x^{[n]}, w_{1},\hat{w}_2,\cdots,\hat{w}_{K}  \mid  [w_1,w_2,\cdots,w_K], y_{1}^{[n]},y_{2}^{[n]},\cdots,y_{K}^{[n]}) \notag \\
	&\hspace{8cm} \times \mathcal{N}_{Y_1\cdots Y_K\mid X}^{\otimes n}(y_{1}^{[n]},y_{2}^{[n]},\cdots,y_{K}^{[n]} \mid x^{[n]}) \\
	&= \frac{1}{\prod_{k=1}^K|\mathcal{M}_k|} \sum_{w_{[K]},x^{[n]},y^{[n]}_{1}}   \mathcal{Z}_{01}(x^{[n]}, w_{1}   \mid  [w_1,w_2,\cdots,w_K], y_{1}^{[n]})  \notag \\
	&\hspace{8cm} \times \sum_{y^{[n]}_{2}, \cdots, y^{[n]}_{K}}  \mathcal{N}_{Y_1\cdots Y_K\mid X}^{\otimes n}(y_{1}^{[n]},y_{2}^{[n]},\cdots,y_{K}^{[n]} \mid x^{[n]}) \\
	&= \frac{1}{\prod_{k=1}^K|\mathcal{M}_k|}  \sum_{w_{[K]}, x^{[n]},y^{[n]}_{1} } \mathcal{Z}_{01}(x^{[n]}, w_{1}   \mid  [w_1,w_2,\cdots,w_K], y_{1}^{[n]})  \times \mathcal{N}_{Y_1\mid X}^{\otimes n}(y_{1}^{[n]}  \mid x^{[n]}) \label{eq:same_marginal_last} \\
	&= \frac{1}{\prod_{k=1}^K|\mathcal{M}_k|}  \sum_{w_{[K]}, x^{[n]},y^{[n]}_{1} } \mathcal{Z}_{01}(x^{[n]}, w_{1}   \mid  [w_1,w_2,\cdots,w_K], y_{1}^{[n]})  \times \widetilde{\mathcal{N}}_{Y_1\mid X}^{\otimes n}(y_{1}^{[n]}  \mid x^{[n]}) \label{eq:use_same_marginal_condition} \\
	&=1-\widetilde{\epsilon}_1
\end{align}
}%
where the same marginal condition is used in \eqref{eq:use_same_marginal_condition}. The last step is because the reasoning leading to \eqref{eq:same_marginal_last} also applies starting from $1-\widetilde{\epsilon}_1$ with the channel $\mathcal{N}_{Y_1\cdots Y_K\mid X}$ replaced by $\widetilde{N}_{Y_1\cdots Y_K\mid X}$. This concludes the proof.	

\section{Proof of Theorem \ref{thm:tree_network}: NS-assisted coding} \label{proof:tree_NS}
Let us first argue that without loss of generality we can consider that for Tx-$b$, $b\in [B]$, there is exactly one Rx that is associated with  Tx-$b$. This is argued as follows. Since all Rxs that are associated with a given Tx are statistically equivalent, by the same-marginals property, they can be treated as one (super)-Rx. In other words, if the message for the super-Rx has $d$ DoF, then this $d$ DoF can be arbitrarily allocated to the Rxs corresponding to the super-Rx. Therefore, henceforth let us assume that for $k\in [K]$, Rx-$k$ is associated with Tx-$k$ and that $B=K$. Note that now we have $d_{\Sigma}(\mbox{\footnotesize Tx-$k$}) = d_k$ for $k\in [K]$.

\subsection{Proof of NS converse}
In this section we show that for NS-assisted coding schemes, $d_k\leq 1$ for $k\in [K]$. Let us first prove it for the $\mathbb{F}_q$ model. For $k\in [K]$, consider the channel to Rx-$k$. According to Lemma \ref{lem:p2p_capacity}, the rate for $W_k$ satisfies $$R_k \leq \max_{P_{X_1\cdots X_K}}I(X_1,\cdots, X_K;Y_k, {\bf G})= \max_{P_{X_1\cdots X_K}}I(X_1,\cdots, X_K;Y_k \mid {\bf G}) \leq \log_2 q$$ as $Y_k\in \mathbb{F}_q$. It follows that $d_k \leq 1, \forall k\in [K]$ if $(d_1,\cdots, d_K) \in \mathcal{D}^{\NS}$.
For the real Gaussian model, Lemma \ref{lem:p2p_capacity}, and the classical DoF result \cite{Jafar_FnT} imply that $d_k \leq 1, \forall k\in [K]$ if $(d_1,\cdots, d_K) \in \mathcal{D}^{\NS}$. \hfill \qed

\subsection{Proof of NS achievability: $\mathbb{F}_q$ model}\label{sec:oneproof}
In this section we  show that the rate tuple $(R_1,R_2,\cdots, R_K) = (\log_2 q, \log_2 q,\cdots, \log_2 q)$ is achievable by NS-assisted coding schemes, which also implies the DoF tuple $(d_1,d_2,\cdots, d_K) = (1,1,\cdots, 1)$ is achievable by NS-assisted coding schemes. 
 
Suppose, given the tree graph $\mathcal{T}$, the indices of the Txs, i.e., $0,1,\cdots, K$, are determined by running depth-first-search (DFS) on $\mathcal{T}$. Recall that for $k\in [K]$,  Rx-$k$ is the Rx that is associated with Tx-$k$. This yields a channel connectivity matrix ${\bf M}$, such that all the elements on its main diagonal are $*$, and all the elements above the main diagonal are zeros. To see this, note that $M_{kk} = *$ as Rx-$k$ is connected to Tx-$k$ by definition. Meanwhile, if $j>i$, then Tx-$j$ cannot be an ancestor of Tx-$i$, as ancestor nodes must appear earlier in a DFS. It follows that Rx-$i$ is not connected to Tx-$j$.
 
Now since the diagonal elements of ${\bf M}$ are non-zeros, for $k\in [K]$, Rx-$k$ can normalize its channel coefficient vector by $G_{kk}$ so that the channel matrix ${\bf G}$ has all $1$'s on the main diagonal after this normalization. The remaining non-zero elements of ${\bf G}$ are still independently uniformly distributed over $\mathbb{F}_q^{\times}$. Thus, henceforth in this section we let the channel coefficient matrix be
\begin{align} \label{eq:G_normalized}
{\bf G}^{(\tau)} = 
\begin{bmatrix}
1&0&\cdots&0\\
G_{21}^{(\tau)} &1&0~\cdots&0\\
\vdots&\ddots&\ddots&\vdots\\
G_{K1}^{(\tau)} &\cdots&G_{K,K-1}^{(\tau)} &1
\end{bmatrix},
\end{align}
and we point out that each $G_{kj}^{(\tau)}$ is independently uniformly distributed over $\mathbb{F}_q^{\times}$ if $M_{kj} = *$, and $G_{kj}^{(\tau)} = 0$ if $M_{kj} = 0$ for $k\in \{2,3,\cdots, K\}, j\in [k-1]$.
 
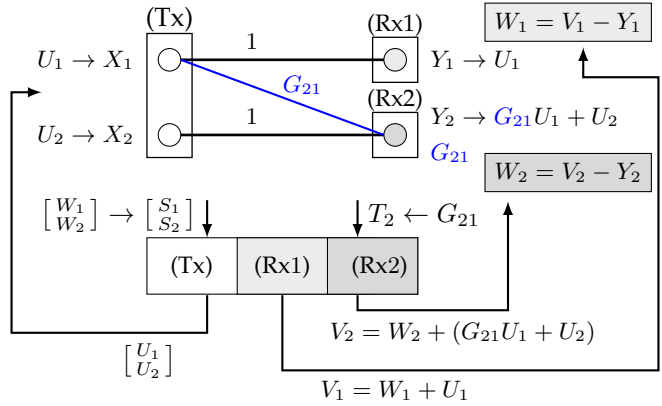
\begin{figure}[!h]
\center
\begin{tikzpicture}
\def \r {0.15}
\def \w {3}
\def \d {1}

\foreach \k in {1,2}{
	\node (T\k) at (0, {-1 * \k * \d}) {};
	\node (R\k) at (\w, {-1 * \k * \d}) {};
}

\foreach \k in {1,2}{
	\draw (T\k) circle (\r) node {};
}

\draw [fill = gray!15] (R1) circle (\r) node {};
\draw [fill = gray!30] (R2) circle (\r) node {};

\draw [line width = 0.5] ({-2*\r},-0.65) rectangle ({2*\r},-2.3);
\foreach \k in {1,2}{
	\draw [line width = 0.5] ({-2*\r + \w}, -1 * \k * \d + 2*\r ) rectangle ({2*\r + \w},-1 * \k * \d - 2*\r);
}
\node [left = 0.2cm of T1] {\footnotesize $U_1 \rightarrow X_1$};
\node [left = 0.2cm of T2] {\footnotesize $U_2\rightarrow X_2$};
\node [right = 0.2cm of R1] {\footnotesize $Y_1 \rightarrow U_1$};
\node [right = 0.2cm of R2, align=left] {\footnotesize $Y_2 \rightarrow  {\color{blue} G_{21}}  U_1+ U_2$\\ \footnotesize ${\color{blue} G_{21}}$};

\node [above=0.1cm of T1] {\small (Tx)};
\node [above=0.05cm of R1] {\footnotesize (Rx1)};
\node [above=0.05cm of R2] {\footnotesize (Rx2)};

\foreach \k in {1,2}{
\draw [line width = 1, black]($(T\k)+(\r,0)$)--($(R\k)-(\r,0)$) node [pos=0.35, above] {\footnotesize $1$};
}
\draw [line width = 0.8, blue]($(T1)+(\r,0)$)--($(R2)-(\r,0)$) node [pos=0.6, above, blue] {\footnotesize $ G_{21}$};

\node [rectangle, minimum height = 0.75cm,  black] (NS) at (1.5, -3.75) {};
\draw[fill=white] ($(NS.north)+(-1.8,0)$) rectangle ($(NS.north)+(1.8,-0.75)$);
\draw[fill=gray!15] ($(NS.north)+(-0.6,0)$) rectangle ($(NS.north)+(1.2,-0.75)$);
\draw[fill=gray!30] ($(NS.north)+(0.6,0)$) rectangle ($(NS.north)+(1.8,-0.75)$);

\node [left=0.65cm of NS] {\footnotesize (Tx)};
\node at (NS) {\footnotesize (Rx$1$)};
\node [right=0.65cm of NS] {\footnotesize (Rx$2$)};

\draw[latex-, thick] ($(NS.north)+(-1,0)$) -- ($(NS.north)+(-1,0.5)$) node [left, pos = 0.6] {\small $\bbsmatrix{W_1\\W_2} \rightarrow \bbsmatrix{S_1\\S_2}$};
\draw[latex-, thick] ($(NS.north)+(1,0)$) -- ($(NS.north)+(1,0.5)$) node [right, pos = 0.6] {\small $T_2 \leftarrow  G_{21} $};

\draw[-latex, black, thick] ($(NS.south)+(-1,0)$) -- ($(NS.south)+(-1,-0.5)$)-- ($(NS.south)+(-3.6,-0.5)$) node [left, pos = 0.3, below] {\small $\bbsmatrix{U_1\\U_2}$} -- ($(NS.south)+(-3.6,2.7)$) -- ($(NS.south)+(-3.2,2.7)$);
\draw[-latex, thick] ($(NS.south)+(0,0)$) -- ($(NS.south)+(0,-1.0)$)--node [below, pos = 0.3]{\footnotesize $V_1=W_1+U_1$} ($(NS.south)+(5,-1)$) -- ($(NS.south)+(5,2.9)$) -- ($(NS.south)+(4,2.9)$)-- ($(NS.south)+(4,3.3)$);

\draw[-latex, thick] ($(NS.south)+(1,0)$) -- ($(NS.south)+(1,-0.2)$)--node [below, pos = 0.7]{\footnotesize $V_2=W_2+(G_{21}U_1+U_2)$}($(NS.south)+(3,-0.2)$) -- ($(NS.south)+(3,1.25)$);

\node at (5.3,-2.5) [draw, rectangle, fill=gray!30]{\footnotesize $W_2 =  V_2-Y_2$};
\node at (5.3,-0.5) [draw, rectangle, fill=gray!15]{\footnotesize $W_1 =  V_1-Y_1$};

\end{tikzpicture}
\caption{$K=2$ case. NS-assisted scheme (shown) achieves capacity $C^{\NS}_\Sigma=2\log_2 q$. In contrast,  the classical capacity $C_\Sigma= \log_2 q+o_q(\log_2(q))$.}
\label{fig:NS_2user}
\end{figure}
The NS coding scheme we present in this section requires only $n=1$ channel-use, allowing us to omit the channel-use indices to simplify notation. The message sets are $\mathcal{M}_k = \mathbb{F}_q$,  $\forall k\in [K]$. 
The scheme requires a $(K+1)$-partite NS box $\mathcal{Z}$, shared across the Tx and the $K$ receivers. Let $S=[S_1,\cdots, S_K]\in \mathbb{F}_q^K$ denote the input at the Tx, and $U = [U_1,\cdots, U_K] \in \mathbb{F}_q^K$ denote the output at the Tx. For $k\in [K]$, let $T_{k} = [T_{k1},\cdots, T_{k,k-1}] \in \mathbb{F}_q^{k-1}$ denote the input to the NS box at Rx-$k$, and $V_k\in \mathbb{F}_q$ denote the output at Rx-$k$. Note that Rx$1$ has only a trivial (constant) input. The box is defined such that, for $S_i =s_i  \in \mathbb{F}_q, \forall i\in [K], T_{j,k} = t_{j,k} \in \mathbb{F}_q, \forall j\in \{2,\cdots, K\}, k\in [j],  U_l = u_l \in \mathbb{F}_q, \forall l\in [K]$ and $V_m =v_m \in \mathbb{F}_q, \forall m\in [K]$,
\begin{align}\label{eq:NSbox_def}
	&\mathcal{Z}\Big([u_1,u_2,\cdots, u_K], ~v_1,~v_2,~\cdots,~v_K \mid \notag \\
	& [s_1,s_2,\cdots, s_K],~ t_{21},~ [t_{31}, t_{32}],~ \cdots,~ [t_{K1},\cdots, t_{K,K-1}]\Big) = \notag \\
	&\begin{cases}
		1/q^K, & \mbox{if}~ \bbsmatrix{v_1\\v_2\\\vdots \\ v_K}
		= \bbsmatrix{s_1\\s_2\\\vdots \\ s_K} +  \bbsmatrix{1&0&\cdots&0\\
	t_{21} &1&0~\cdots&0\\
	\vdots&\ddots&\ddots&\vdots\\
	t_{K1} &\cdots & t_{K,K-1} &1} 
	\bbsmatrix{u_1\\u_2\\\vdots \\u_K}\\
	0 & \mbox{otherwise.}
	\end{cases}.
\end{align}
Before proving that the box is non-signaling, let us first explain how the box is utilized. At the Tx let $[S_1,\cdots, S_K] = [W_1,\cdots, W_K]$. The box outputs $[U_1,\cdots, U_K]$. The Tx sends to the channel $[X_1,\cdots, X_K] = [U_1,\cdots, U_K]$. For $k\in [K]$, each Rx-$k$ sets $[T_{k1}, \cdots, T_{k,k-1}] = [G_{k1}, \cdots, G_{k,k-1}]$, so that the outputs of the NS box at these receivers are (written collectively),
\begin{align} \label{eq:NSbox_output_Rxs}
	\begin{bmatrix}
		V_1\\V_2\\\vdots \\ V_K
	\end{bmatrix}
	&=
	\begin{bmatrix}
		W_1\\W_2\\\vdots \\ W_K
	\end{bmatrix}
	+
	\underbrace{\begin{bmatrix}
	1&0&\cdots&0\\
	G_{21} &1&0~\cdots&0\\
	\vdots&\ddots&\ddots&\vdots\\
	G_{K1} &\cdots & G_{K,K-1} &1
	\end{bmatrix}
	\begin{bmatrix}
		U_1\\U_2\\\vdots \\U_K
	\end{bmatrix}}_{[Y_1,\cdots, Y_K]^\top},
\end{align}
where we made the observation that $V_k=W_k+Y_k, \forall k\in [K]$ with probability $1$, according to \eqref{eq:NSbox_def}.
The decoding at Rx-$k$ is done by subtracting $Y_k$ from $V_k$ since $W_k = V_k - Y_k$ for $k\in [K]$. Therefore, the rate tuple $(\log_2 q, \log_2 q \cdots, \log_2 q)$ is achievable by NS-assisted coding schemes. 

We now show that $\mathcal{Z}$ is non-signaling. According to \cite{barrett2005nonlocal, masanes2006general}, it suffices to verify that $$\sum_{(u_1,\cdots, u_K) \in \mathbb{F}_q^K}\mathcal{Z}(\cdot \mid \cdot)$$ does not depend on $(s_1,\cdots, s_K)$, and that for $k\in [K]$, $$\sum_{v_k \in \mathbb{F}_q}\mathcal{Z}(\cdot \mid \cdot)$$ does not depend on $(t_{k1}, \cdots, t_{k,k-1})$, where $\mathcal{Z}(\cdot \mid \cdot)$ is the shorthand notation for the box distribution \eqref{eq:NSbox_def}. The first condition is verified, because $\sum_{(u_1,\cdots, u_K) \in \mathbb{F}_q^K}\mathcal{Z}(\cdot \mid \cdot) = 1/q^K$ which follows from the following reason: Given $\{s_k\}, \{t_{k,i}\}$ and $\{v_k\}$, since the lower triangular matrix composed of $\{t_{k,i}\}$ in the first condition of \eqref{eq:NSbox_def} always has full rank, there is a unique $[u_1,\cdots, u_K]\in \mathbb{F}_q^K$ for which the first condition of \eqref{eq:NSbox_def} is satisfied. Therefore, only one $[u_1,\cdots, u_K]\in \mathbb{F}_q^K$ yields $\mathcal{Z}(\cdot \mid \cdot) =1/q^K$, and the others terms are equal to $0$.  Next, for $k\in [K]$, given $\{s_k\}, \{t_{k,i}\}$, $\{u_k\}$ and $\{v_{k'}\}_{k'\in [K]\setminus \{k\}}$
\begin{align}
	&\sum_{v_k\in \mathbb{F}_q}\mathcal{Z}(\cdot \mid \cdot)=  \begin{cases}
		1/q^K & ~\mbox{if}~v_{k'} = s_{k'}+\sum_{i=1}^{k'-1}t_{k'i}u_i + u_{k'}, \forall k'\not=k,  \\
		0 &  \mbox{otherwise,}
	\end{cases}\notag
\end{align}
which does not depend on $[t_{k1}, \cdots, t_{k,k-1}]$. This concludes the proof that $\mathcal{Z}$ is non-signaling. \hfill \qed

\subsection{Alternative proof: Successive encoding using bipartite NS boxes}\label{sec:altproof}

The scheme in this subsection requires in total $K-1$ \emph{bipartite} NS boxes, denoted as $\mathcal{Z}_2,\mathcal{Z}_3,\cdots, \mathcal{Z}_K$, with the $k^{th}$ NS box $\mathcal{Z}_k$ shared between the Tx and Rx-$k$, for all $k\in \{2,3,\cdots,K\}$. To serve as visual aids for the following description of the scheme, in addition to the scheme for $K=2$ in Figure \ref{fig:NS_2user_bipartite}, let us provide an explicit solution for $K=3$ in Figure \ref{fig:K3}.

\begin{figure}[!h]
\center
\begin{align*}
&&&\begin{array}{ll}
X_1&=W_1\\
Y_1&=X_1\\
&=W_1
\end{array}\\
\begin{array}{r}\mbox{(Tx, Rx-$2$)}\\
\mbox{NS Box $\mathcal{Z}_2$}
\end{array}  \begin{array}{|c|c|c|}\hline
 \mbox{Party}&\mbox{Tx}&\mbox{Rx-$2$}\\\hline
\mbox{inputs}&X_1&G_{21}\\\hline
\mbox{outputs}&U_2&{\color{blue}G_{21}X_1-U_2}\\\hline
\end{array} &&&
\begin{array}{ll} X_2&=W_2-U_2\\ Y_2&=G_{21}X_1+X_2\\ &={\color{blue}G_{21}X_1-U_2}+W_2 \end{array}\\
\begin{array}{r}\mbox{(Tx, Rx-$3$)}\\
\mbox{NS Box $\mathcal{Z}_3$}
\end{array}  \begin{array}{|c|c|c|}\hline
\mbox{Party}&\mbox{Tx}&\mbox{Rx-$3$}\\\hline
\mbox{inputs}&X_1,X_2&G_{31},G_{32}\\\hline
\mbox{outputs}&U_3&{\color{blue}G_{31}X_1+G_{32}X_2-U_3}\\\hline
\end{array} &&& 
\begin{array}{ll} X_3&=W_3-U_3\\ Y_3&=G_{31}X_1+G_{32}X_2+X_3\\ &={\color{blue}G_{31}X_1+G_{32}X_2-U_3}+W_3\end{array}
\end{align*}
\caption{$K=3$ case, with the use of $K-1=2$ bipartite NS boxes.}\label{fig:K3}
\end{figure}

Let $(S_k, U_k)$ denote the inputs and $(T_k, V_k)$ denote the outputs of $\mathcal{Z}_k$. Note that $(S_k, U_k)_{k\in \{2,\cdots, K\}}$ are with the Tx, whereas $(T_k,V_k)$ is with Rx-$k$ for $k\in [K]$. For vectors $a = (a_1,a_2,\cdots, a_m)\in \mathbb{F}_q^m$ and $b = (b_1,b_2,\cdots, b_m)\in \mathbb{F}_q^m$, we denote $a \cdot b \triangleq \sum_{i=1}^m a_ib_i$ as the ($\mathbb{F}_q$) `inner-product' between $a$ and $b$. For $k\in \{2,3,\cdots, K\}$, $\mathcal{Z}_k$ is defined over input alphabets $\mathcal{S}_{k} = \mathcal{T}_k = \mathbb{F}_q^{k-1}$, output alphabets $\mathcal{U}_{k} = \mathcal{V}_k = \mathbb{F}_q$, and is specified as
\begin{align} \label{eq:inner_product_box}
	&\mathcal{Z}_k(u, v \mid s,t) =
	\begin{cases}
		1/q, & u+v = s\cdot t\\
		0, & u+v \not= s\cdot t
	\end{cases},\\
	& \hspace{3cm} \forall s \in \mathbb{F}_q^{k-1}, t\in \mathbb{F}_q^{k-1}, u\in \mathbb{F}_q, v\in \mathbb{F}_q. \notag
\end{align}
This is also an OTP-box described in \cite{OTPmodel} and therefore it is non-signaling. Note that one can also construct $\mathcal{Z}_k$ by adding the outputs of $k-1$ $\mathbb{F}_q$ PR boxes, similar to the approach used in the van Dam protocol \cite{van2013implausible}. 

Algorithm \ref{alg:finite_field} specifies the inputs to the NS boxes and to the channel.

\begin{algorithm}
\caption{Successive encoding using bipartite NS boxes}\label{alg:finite_field}
\begin{algorithmic}

\State $X_1 \gets W_1$ 
\For{$k \gets 2, 3,\cdots, K$}
	\State $S_k \gets (X_1,X_2,\cdots, X_{k-1})$  \Comment{(Tx obtains $U_k$)}                   
    \State $X_k \gets W_k-U_k$
\EndFor
\algrule
\For{$k \gets 2,3,\cdots, K$}
	\State $T_k \gets (G_{k1}, G_{k2}, \cdots, G_{k,k-1})$ \Comment{(Rx-$k$ obtains $V_k$)}
\EndFor
\end{algorithmic}
\end{algorithm}
The decoding at Rx$1$ is direct as it sees $Y_1=X_1 = W_1$ from the channel. For $k\in \{2,3,\cdots, K\}$, Rx-$k$ subtracts $V_k$ from $Y_k$ and obtains
\begin{align}
	&Y_k-V_k = \underbrace{G_{k1}X_1+G_{k2}X_2+\cdots+ G_{k,k-1}X_{k-1}}_{=S_k\cdot T_k} + X_k - V_k \\
	&= S_k\cdot T_k + W_k - (U_k+V_k) \\
	& = W_k
\end{align}
with certainty, since \eqref{eq:inner_product_box} guarantees $U_k+V_k = S_k \cdot T_k$. Therefore, the rate tuple $(\log_2 q, \log_2 q \cdots, \log_2 q)$ is achievable by NS-assisted coding schemes. 
This proves $C^{\rm NS}\geq K\log_2 q$.

 \hfill \qed

 \begin{remark}[Proof of Theorem \ref{thm:triangle_achi}]\label{rem:triangle_number}
	The NS-assisted coding scheme works as long as ${\bf M}$ has only $\ast$ on the main diagonal. In other words, the scheme works even if $G_{kj}=0$ for some $k\in [K], j\in [k-1]$. Suppose ${\bf M}$ contains a $D\times D$ sub-matrix ${\bf M}'$ which, upon row and column permutations yields a $D\times D$ lower triangle matrix, then for those receivers (say indexed by $\{i_1,i_2,\cdots, i_D\}$) corresponding to the rows of the submatrix, the proof implies that $d_{k} = 1$ is simultaneously achievable for all $k\in \{i_1,i_2,\cdots, i_D\}$ by NS-assisted coding schemes, when considering only the sub-network ${\bf M}'$.
\end{remark}

\subsection{Proof of NS achievability:  Gaussian model}
We again consider the normalized channel matrix ${\bf G}$ as in \eqref{eq:G_normalized}. The  construction is based on the bipartite NS boxes construction for the  $\mathbb{F}_q$ model.
\newcommand{\stP}{\lceil\sqrt{P}\rceil}

Given the power constraint $P$ and for $k\in \{2,3,\cdots, K\}$, define the NS box
\begin{align}
	&\mathcal{Z}_k(u,v\mid s,t) = \begin{cases}
		\frac{1}{\stP},  &u+v = \lfloor  s \cdot t \rfloor \mod \stP \\
		0, & u+v \not= \lfloor  s \cdot t \rfloor \mod \stP
	\end{cases},  \\
	& \hspace{2cm} \forall \notag s\in \mathbb{R}^{k-1}, t\in \mathbb{R}^{k-1}, u\in \{0,1,\cdots, \stP-1\}, v\in \{0,1,\cdots, \stP-1\}.
\end{align}
The box being NS can be verified as $\mathcal{Z}_k(u\mid s,t) = \mathcal{Z}_k(v \mid s,t) = 1/\stP$ for any $u,v, s,t$.

In the following algorithm, for each use of the real Gaussian channel with (input, output) $=$ $(\{X_k\}, \{Y_k\})$, we convert it to another channel with (input, output) denoted as $(\{\overline{X}_k, \overline{Y}_k\})$, where $\overline{X}_k \in \{0,1,\cdots, \stP-1\}$ and $\overline{Y}_k \in [0,\stP)$ for $k\in [K]$. Since the same algorithm works for every channel use, we omit the channel-use indices.

\begin{algorithm}
\caption{Channel conversion}\label{alg:guassian}
\begin{algorithmic}
\State $X_1 \gets \overline{X}_1$ 
\For{$k \gets 2, 3,\cdots, K$}
	\State $S_k \gets (X_1,X_2,\cdots, X_{k-1})$  \Comment{(Tx obtains $U_k$ from $\mathcal{Z}_k$)}                   
    \State $X_k \gets \overline{X}_k-U_k \mod \stP$ 
\EndFor
\algrule
\For{$k \gets 2,3,\cdots, K$}
	\State $T_k \gets (G_{k1}, G_{k2}, \cdots, G_{k,k-1})$ \Comment{(Rx-$k$ obtains $V_k$ from $\mathcal{Z}_k$)}
\EndFor
\algrule
\State $\overline{Y}_1 \gets Y_1 \mod \stP$ 
\For{$k \gets 2,3,\cdots, K$}
	\State $\overline{Y}_k \gets (Y_k - V_k) \mod \stP$
\EndFor \Comment{(Rx-$k$ obtains $\overline{Y}_k, \forall k\in [K]$) }
\end{algorithmic}
\end{algorithm}
Note that the input to the channel, $X_k$, is always in $\{0,1,\cdots, \stP-1\}$, thus the power constraint is satisfied.
Following Algorithm \ref{alg:guassian}, we obtain that
\begin{align}
	\overline{Y}_1 = \overline{X}_1 + Z_1 \mod \stP
\end{align}
and that for $k\in \{2,3,\cdots, K\}$,
\begin{align}
	\overline{Y}_k &= Y_k - V_k \mod \stP \\
	&= G_{k1}X_1+\cdots+G_{k,k-1}X_{k-1} + X_k + Z_k - V_k \mod \stP \\
	&= \overline{X}_k + Z_k + G_{k1}X_1+\cdots+G_{k,k-1}X_{k-1}  - U_k- V_k  \mod \stP \\
	&= \overline{X}_k + Z_k + G_{k1}X_1+\cdots+G_{k,k-1}X_{k-1}  - \lfloor S_k \cdot T_k \rfloor \mod \stP \\
	&= \overline{X}_k + Z_k +  S_k \cdot T_k  - \lfloor S_k \cdot T_k \rfloor \mod \stP \\
	&\triangleq \overline{X}_k + \underbrace{Z_k + \widetilde{Z}_k}_{\overline{Z}_k} \mod \stP
\end{align}
where $\widetilde{Z}_k$ is a random variable distributed over $[0,1)$. 
The variance of $\widetilde{Z}_k$ is upper bounded by $1/4$, 
the variance of $Z_k$ is upper bounded by $1$,
and thus the variance of $\overline{Z}_k$ is upper bounded by $9/4$. Now for the converted channel, for $k\in [K]$, we obtain that if $P_{\overline{X}_k}$ is the uniform distribution over $\{0,1,\cdots, \stP-1\}$,
\begin{align}
	&I(\overline{X}_k; \overline{Y}_k) \notag \\
	&\geq I(\overline{X}_k; \lfloor \overline{Y}_k \rfloor) \\
	&= H(\lfloor \overline{Y}_k \rfloor) - H(\lfloor \overline{Y}_k \rfloor \mid \overline{X}_k) \\
	&= \log_2 \stP- H(\lfloor \overline{Y}_k \rfloor \mid \overline{X}_k)
\end{align}
and that
\begin{align}
	&H(\lfloor \overline{Y}_k \rfloor \mid \overline{X}_k) \notag \\
	&=H(\overline{X}_k + \lfloor \overline{Z}_k \rfloor \mod \stP \mid \overline{X}_k) \label{eq:achi_floor_mod} \\
	&\leq H(\lfloor \overline{Z}_k \rfloor) \\
	&= o_P(\log_2 P) \label{eq:achi_bounded_variance}
\end{align}
where Step \eqref{eq:achi_floor_mod} is because $\lfloor a \mod d\rfloor = \lfloor a \rfloor \mod d$ and $\overline{X}_k$ is an integer. Step \eqref{eq:achi_bounded_variance} is because $\overline{Z}_k$ has bounded variance (which does not depend on $P$)  and for a discrete random variable with bounded variance, the entropy is also upper bounded (by a constant that does not depend on $P$) \cite{agostini2019discrete}. 
Therefore, for $k\in [K]$, we have a converted channel that has input $\overline{X}_k$ and output $\overline{Y}_k$ at Rx-$k$. Note that these $K$ converted channels operate independently. It follows that the DoF for the $k^{th}$ message, $$d_k^{\NS} = \lim_{P\to \infty}\frac{\log_2 \stP - o_P(\log_2 P)}{\frac{1}{2}\log_2 P} = 1,$$ is achievable by NS-assisted coding schemes simultaneously for all $k\in [K]$. This completes the proof. \hfill \qed

\section{Proof of Theorem \ref{thm:tree_network}: Classical coding} \label{proof:tree_classical}
Let us again assume that $B=K$ and that there is a unique Rx-$k$ that is associated with Tx-$k$ for $k\in [K]$, as is argued in Appendix \ref{proof:tree_NS} for NS-assisted coding schemes.
\subsection{Proof of classical converse: AIS bound}\label{sec:AIS}
The high-level idea of the proof is to identify a subset of transmit antennas and receivers with connectivity as shown in Fig. \ref{fig:KuserMISOBC_finite_field}, and to prove that the sum-DoF value for this subset of receivers is upper bounded by $1$, by adapting the Aligned Image Sets bound in \cite{davoodi2016aligned}, which was previously established for a wireless network with the same connectivity. The bound thus obtained corresponds to a bound specifying the region $\mathcal{D}$ in \eqref{eq:tree_region_classical}.
Let $(\mbox{Tx-}0, \mbox{Tx-}b_1, \cdots, \mbox{Tx-}b_{L})$ be any root-to-leaf path of the tree graph that describes the channel connectivity. In this subsection we shall show that $d_{b_1} + d_{b_2} + \cdots + d_{b_L} \leq 1$. 
Without loss of generality, let $b_1 = 1, b_2=2, \cdots, b_{L} = L$.  Consider only the first $L$ receivers, Rx-$1$, Rx-$2$, $\cdots$, Rx-$L$. Since these $L$  receivers  correspond to Txs that  lie on a root-to-leaf path, they  are only connected to  transmit antennas indexed by $(1,2,\cdots, L)$. Specifically, consider the connectivity matrix ${\bf M}'\in \{0,*\}^{L \times L}$ corresponding to the $L$ receivers and the $L$ transmit antennas, which is
\begin{align} \label{eq:triangular_M}
	{\bf M}'  = \begin{bmatrix}
		\ast & 0 & \cdots & 0 \\ \ast & \ast & \cdots & 0 \\ \vdots	& \vdots & \ddots & \vdots \\ \ast & \ast & \cdots & \ast
	\end{bmatrix}^{L \times L}.
\end{align}
Note that these $L$ transmit antennas and the $L$ receivers together form the following reduced CoMP BC network, as illustrated in Figure \ref{fig:KuserMISOBC_finite_field}.
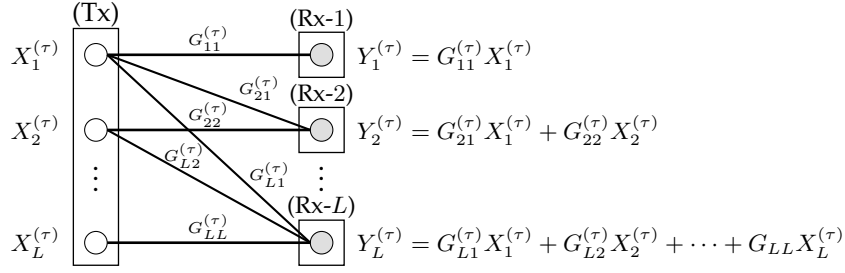
\begin{figure}[!thbp]
\center
\begin{tikzpicture}
\def \r {0.15}
\def \w {3}
\def \d {1}

\foreach \k in {1,2}{
	\node (T\k) at (0, {-1 * \k * \d}) {};
	\node (R\k) at (\w, {-1 * \k * \d}) {};
}
\node (T3) at (0, {-1 * 3.5 * \d}) {};
\node (R3) at (\w, {-1 * 3.5 * \d}) {};
\foreach \k in {1,2}{
	\draw (T\k) circle (\r) node {};
	\draw [fill = gray!25] (R\k) circle (\r) node {};
}
\draw (T3) circle (\r) node {};
\draw [fill = gray!25] (R3) circle (\r) node {};
\draw [line width = 0.5] ({-2*\r},-0.65) rectangle ({2*\r},-3.8);
\foreach \k in {1,2}{
	\draw [line width = 0.5] ({-2*\r + \w}, -1 * \k * \d + 2*\r ) rectangle ({2*\r + \w},-1 * \k * \d - 2*\r);
}
\draw [line width = 0.5] ({-2*\r + \w}, -1 * 3.5 * \d + 2*\r ) rectangle ({2*\r + \w},-1 * 3.5 * \d - 2*\r);
\foreach \k in {1,2}{
	\node [left = 0.2cm of T\k] {\footnotesize $X_\k^{(\tau)}$};
}
\node [left = 0.2cm of T3] {\footnotesize $X_{L}^{(\tau)}$};
\node [right = 0.2cm of R1] {\footnotesize $Y_1^{(\tau)} = G_{11}^{(\tau)}X_1^{(\tau)}$};
\node [right = 0.2cm of R2] {\footnotesize $Y_2^{(\tau)} = {G_{21}^{(\tau)}}X_1^{(\tau)}+G_{22}^{(\tau)}X_2^{(\tau)}$};
\node [right = 0.2cm of R3, align=center] {\footnotesize $Y_{L}^{(\tau)} = { G_{{L}1}^{(\tau)}}X_1^{(\tau)}+{ G_{{L}2}^{(\tau)}}X_2^{(\tau)}+ \cdots + G_{{L}{L}} X_{L}^{(\tau)}$};
\node [below = 0cm of T2] {$\vdots$};
\node [below = 0cm of R2] {$\vdots$};

\node [above=0.1cm of T1] {\small (Tx)};
\node [above=0.05cm of R1] {\footnotesize (Rx-$1$)};
\node [above=0.05cm of R2] {\footnotesize (Rx-$2$)};
\node [above=0.05cm of R3] {\footnotesize (Rx-$L$)};
 
\draw [line width = 1, black]($(T1)+(\r,0)$)--($(R1)-(\r,0)$) node [pos=0.5, above=-0.1] {\tiny $G_{11}^{(\tau)}$};
\draw [line width = 1, black]($(T2)+(\r,0)$)--($(R2)-(\r,0)$) node [pos=0.5, above=-0.1] {\tiny $G_{22}^{(\tau)}$};
\draw [line width = 1, black]($(T3)+(\r,0)$)--($(R3)-(\r,0)$) node [pos=0.5, above=-0.1] {\tiny $G_{L L}^{(\tau)}$};

\draw [line width = 0.8, black]($(T1)+(\r,0)$)--($(R2)-(\r,0)$) node [pos=0.75, above, black] {\tiny $G_{21}^{(\tau)}$};
\draw [line width = 0.8, black]($(T1)+(\r,0)$)--($(R3)-(\r,0)$) node [pos=0.8, above=0.1, black] {\tiny $G_{{L}1}^{(\tau)}$};
\draw [line width = 0.8, black]($(T2)+(\r,0)$)--($(R3)-(\r,0)$) node [pos=0.38, above=-0.08, black] {\tiny $G_{{L}2}^{(\tau)}$};

\end{tikzpicture}
\caption{Reduced CoMP BC network for the  $L$ transmit antennas and receivers.}
\label{fig:KuserMISOBC_finite_field}
\end{figure}

Reference \cite{davoodi2016aligned} shows that for the $L$-user MISO BC channel with the connectivity shown in Figure \ref{fig:KuserMISOBC_finite_field}, the sum-DoF is upper bounded by $1$ for the Gaussian model, thus proving the desired bound $d_1+d_2+\cdots+d_L \leq 1$. In the following we prove the same bound for the $\mathbb{F}_q$ model, by adapting the proof of \cite{davoodi2016aligned} to the $\mathbb{F}_q$ model.

Given any achievable rate tuple $(R_1,\cdots, R_{L})$ with classical coding schemes, there exists a sequence (indexed by $n$) of classical coding schemes satisfying  \eqref{eq:criteria1},\eqref{eq:criteria2}  where the $n^{th}$ scheme has message sets $(\mathcal{M}_1^{(n)}, \cdots, \mathcal{M}_{L}^{(n)})$. Fano's inequality implies,
	$\log_2|\mathcal{M}_k^{(n)}|  \leq I(W_k;\widehat{W}_k) + o_n(n), ~\forall k\in [{L}].$
Therefore, for the desired bound, it suffices to show
\begin{align} \label{eq:lim_n_sum_rate}
	\lim_{n\to \infty} \frac{1}{n}\sum_{k=1}^{L} I(W_k;\widehat{W}_k) \leq \log_2 q + o_q(\log_2 q).
\end{align}
\noindent In the following, $A^{[n]}+B^{[n]} = [A^{(1)}+B^{(1)},\cdots, A^{(n)}+B^{(n)}]$ and $A^{[n]} B^{[n]} = [A^{(1)}B^{(1)},\cdots, A^{(n)}B^{(n)}]$.  We will make frequent use of the property that in any linear combination of entropies, e.g., $\mathcal{L}=H(A|B,V)+H(C|V)-H(D|V)$, all of which include some random variable $V$ in their conditioning, there exist realizations of $V$, say $v_1,v_2$ such that by fixing $V$ at those realizations we obtain, respectively, a lower and an upper bound on $\mathcal{L}$, i.e., 
\begin{align}
\mathcal{L}\geq H(A|B,V=v_1)+H(C|V=v_1)-H(D|V=v_1),\label{eq:geqH}\\
\mathcal{L}\leq  H(A|B,V=v_2)+H(C|V=v_2)-H(D|V=v_2).\label{eq:leqH}
\end{align} 
The  property holds simply because in any average, there must exist an instance that is not smaller than the average, and an instance that is not larger than the average.  

\vspace{0.3cm}

\noindent{\it [Difference of conditional entropies]}:
\begin{align}
	&\sum_{k=1}^{L} I(W_k;\widehat{W}_k) \leq \sum_{k=1}^{L} I(W_k; Y_k^{[n]}, {\bf G}^{[n]}) \label{eq:conv_data_processing} \\
	&  \leq \sum_{k=1}^{L} I(W_k; Y_k^{[n]} \mid {\bf G}^{[n]}, W_{k+1},\cdots, W_{L})   \label{eq:conv_G_ind_W} \\
	&\leq n\log_2 q + \sum_{k=2}^{L} \Big(H(Y_{k-1}^{[n]}\mid {\bf G}^{[n]}, W_k,\cdots, W_{L})  - H(Y_{k}^{[n]}\mid {\bf G}^{[n]}, W_k, \cdots, W_{L})\Big).\label{eq:conv_alphabet_bound}
\end{align}
Step \eqref{eq:conv_data_processing} is by the data processing inequality. Step \eqref{eq:conv_G_ind_W} is because $W_1, W_2, \cdots, W_{L}$ and ${\bf G}^{[n]}$ are mutually independent. \eqref{eq:conv_alphabet_bound} is because $H(Y_{L}^{[n]})\leq n\log_2 q$. Similar to \cite{davoodi2016aligned}, we argue that for each $k\in \{2,3,\cdots,{L}\}$, there exists $(w_k,\cdots, w_{{L}}) \in \mathcal{M}_k\times \cdots \times \mathcal{M}_{L}$ such that
\begin{align}
	&  H(Y_{k-1}^{[n]}\mid {\bf G}^{[n]}, W_k,\cdots, W_{L}) - H(Y_{k}^{[n]}\mid {\bf G}^{[n]}, W_k, \cdots, W_{L})  \notag\\
	&\leq H(Y_{k-1}^{[n]} \mid {\bf G}^{[n]}, W_k=w_k,\cdots, W_{L}=w_{L}) \notag \\
	&\hspace{1cm}- H(Y_k^{[n]} \mid {\bf G}^{[n]}, W_k=w_k,\cdots, W_{L}=w_{L}) \hspace{1cm}\mbox{($\because$ \eqref{eq:leqH})}\\
	&\leq \underbrace{\max  \big( H(Y_{k-1}^{[n]} \mid {\bf G}^{[n]}) - H(Y_{k}^{[n]} \mid {\bf G}^{[n]}) \big)}_{\triangleq \Delta_k} \label{eq:conv_difference}
\end{align}
where for the term $\Delta_k$, the maximum is taken over all distributions $P_{X_1^{[n]}\cdots X_L^{[n]}}$ defined on $(\mathbb{F}_q^n)^L$. 
We will prove that {\bf for each} $k\in \{2,3,\cdots, {L}\}$, we have 
\begin{align} \label{eq:conv_to_prove_Delta_k}
	\Delta_k \leq n o_q (\log_2 q)
\end{align}
so that we conclude $\sum_{k=1}^{L} I(W_k;\widehat{W}_k)/n$ $ \leq  \log_2q +  o_q (\log_2 q)$, and thus prove \eqref{eq:lim_n_sum_rate}.

\vspace{0.3cm}
 
Define $\overline{\bf G}^{[n]} \triangleq ({\bf G}_1^{[n]},\cdots, {\bf G}_{k-1}^{[n]}, {\bf G}_{k+1}^{[n]}, \cdots, {\bf G}_{L}^{[n]})$ as the collection of channel coefficients except for those of Rx-$k$. By \eqref{eq:leqH}, there exists $\overline{\bf g}^{[n]}$ such that, 
\begin{align}
	\eqref{eq:conv_difference} \leq \max_{P_{X_1^{[n]}\cdots X_{L}^{[n]}}} \big( H(Y_{k-1}^{[n]} \mid {\bf G}_k^{[n]}, \overline{\bf G}^{[n]}=\overline{\bf g}^{[n]}) - H(Y_{k}^{[n]} \mid {\bf G}_k^{[n]}, \overline{\bf G}^{[n]} = \overline{\bf g}^{[n]}) \big). \label{eq:conv_difference_alter}
\end{align}
We will proceed conditioned on the event $\overline{\bf G}^{[n]} = \overline{\bf g}^{[n]}$ for the remainder of the proof. Equivalently, for the sake of a compact notation, $\overline{\bf G}^{[n]} = \overline{\bf g}^{[n]}$ is treated as a constant in the remainder of the proof, i.e., the conditioning on the event $\overline{\bf G}^{[n]} = \overline{\bf g}^{[n]}$ will no longer be explicitly specified. Note that this also means that we allow the input distribution $P_{X_1^{[n]}\cdots X_{L}^{[n]}}$  to be optimized for this particular realization $\overline{\bf G}^{[n]} = \overline{\bf g}^{[n]}$. Intuitively, this amounts to giving the Tx the knowledge of the realization of these coefficients, which it can use to optimize its coding scheme. What is crucial is that the channel coefficients associated with Rx-$k$ remain random and unknown to the Tx. 

Since $\overline{\bf G}^{[n]} = \overline{\bf g}^{[n]}$ is determined, note that $Y_{k-1}^{[n]}$ is now a known function of $X_1^{[n]},\cdots, X_{L}^{[n]}$. In general for random variables $A,B$, if $A$ is a known function of $B$, say $A=f(B)$, then there is a one-to-one correspondence between the distribution $P_B$ and  $P_{A,B}=P_AP_{B\mid A}$. Furthermore, by functional representation lemma \cite[Page 626]{NIT}, there exists a function $\phi$ such that the distribution $P_{B\mid A}$ can be simulated as $B=\phi(A,\Xi)$, where $\Xi\sim\mbox{Uniform}(0,1)$ is independent of $A$. Optimizing over the distribution $P_B$ is then equivalent to optimizing over $(P_A,\phi)$. Applying this principle to our setting with $A=Y_{k-1}^{[n]}$ and $B=(X_1^{[n]},\cdots, X_{L}^{[n]})$, there exists a  function $\phi$ such that $(X_1^{[n]},\cdots, X_{L}^{[n]})=\phi(Y_{k-1}^{[n]}, \Xi)$. Equivalently, there exist functions $\phi_l, l\in[L]$ such that $X_l^{[n]} = \phi_l(Y_{k-1}^{[n]}, \Xi)$. Optimizing over $P_{X_1^{[n]}\cdots X_{L}^{[n]}}$ is now equivalent to optimizing over $(P_{Y_{k-1}^{[n]}},\phi_1,\cdots,\phi_L)$. With this representation, we proceed,
\begin{align}
	&\eqref{eq:conv_difference_alter}  \leq \max_{P_{Y_{k-1}^{[n]}},\{\phi_l\}}  \big( H(Y_{k-1}^{[n]}\mid {\bf G}_k^{[n]})  - H(Y_{k}^{[n]}\mid {\bf G}_k^{[n]})\big) \label{eq:conv_deterministic_mappings_start}  \\
	&\leq \max_{P_{Y_{k-1}^{[n]}},\{\phi_l\}} \big( H(Y_{k-1}^{[n]} \mid {\bf G}_k^{[n]} ) -   H( Y_k^{[n]} \mid {\bf G}_k^{[n]}, \Xi )\big) \\
	&\leq \max_{P_{Y_{k-1}^{[n]}},\{\phi_l\}} \big( H(Y_{k-1}^{[n]} \mid {\bf G}_k^{[n]} ) -   H( Y_k^{[n]} \mid {\bf G}_k^{[n]}, \Xi=\xi^* )\big), ~~\exists \xi^* \in [0,1]\hspace{1cm}\mbox{($\because$ \eqref{eq:geqH})}\label{eq:existsxi} \\
	&= \max_{P_{Y_{k-1}^{[n]}},\{\psi_l\}} \big( H(Y_{k-1}^{[n]} \mid {\bf G}_k^{[n]} ) -   H( Y_k^{[n]} \mid {\bf G}_k^{[n]} )\big) \label{eq:conv_deterministic_mappings} \\
	&\leq \max_{P_{Y_{k-1}^{[n]}},\{\psi_l\}} H(Y_{k-1}^{[n]}\mid Y_{k}^{[n]}, {\bf G}_k^{[n]}) \\
	&=\max_{P_{Y_{k-1}^{[n]}},\{\psi_l\}} H\big(Y_{k-1}^{[n]} \mid \underbrace{G_{k1}^{[n]}X_1^{[n]} + \cdots + G_{kk}^{[n]} X_{k}^{[n]}}_{ Y_k^{[n]}\triangleq~\chi(Y_{k-1}^{[n]}, {\bf G}_k^{[n]})}, {\bf G}_k^{[n]}\big) \label{eq:def_Yk_function}
\end{align}
In Step \eqref{eq:conv_deterministic_mappings} we define $X_l^{[n]} = \phi_l(Y_{k-1}^{[n]}, \xi^*)\triangleq \psi_l(Y_{k-1}^{[n]})$ for $l\in [{L}]$. Since all $X_l^{[n]}$ are just functions of $Y_{k-1}^{[n]}$, note that $Y_k^{[n]}$ is now a function of $(Y_{k-1}^{[n]},{\bf G}_k^{[n]})$, as we note explicitly in \eqref{eq:def_Yk_function}.
Note that the function $\chi\colon \mathbb{F}_q^n\times ({\mathbb{F}_q^\times}^n)^k \to \mathbb{F}_q^n$ in \eqref{eq:def_Yk_function} is defined as 
\begin{align}
	\chi(y_{k-1}^{[n]}, {\bf g}_k^{[n]}) = g_{k1}^{[n]}\psi_1(y_{k-1}^{[n]}) + \cdots + g_{kk}^{[n]}\psi_k(y_{k-1}^{[n]}). \label{eq:def_chi_function}
\end{align}
In words, $\chi$ specifies how $Y_{k}^{[n]}$ depends on $Y_{k-1}^{[n]}$ and ${\bf G}_k^{[n]}$.

\vspace{0.3cm}

\noindent {\it [Aligned image sets]}: 
For $y_{k-1}^{[n]}\in \mathbb{F}_q^n$ and ${\bf g}_k^{[n]} = (g_{k1}^{[n]},\cdots, g_{kk}^{[n]}) \in ({\mathbb{F}_q^\times}^n)^{k}$,  define the aligned image set (AIS) \cite{davoodi2016aligned} as
\begin{align} \label{eq:def_AIS}
	&\mathcal{S} (y_{k-1}^{[n]}, {\bf g}_k^{[n]}) \triangleq \big\{\gamma \in \mathbb{F}_q^n \colon \chi(\gamma, {\bf g}_k^{[n]}) = \chi(y_{k-1}^{[n]}, {\bf g}_k^{[n]}) \big\}.
\end{align}

\noindent To continue, let us prove the following lemma on conditional entropy.
\begin{lemma} \label{lem:cond_entropy_on_func}
For a random variable $A$ defined over alphabet $\mathcal{A}$ with distribution $P_A$, we have
	\begin{align}
		H(A\mid f(A)) = \sum_{a\in \mathcal{A}} P_{A}(a) \times H\Big(A~\Big|~ f(A) = f(a)\Big),
	\end{align}
	where $f\colon \mathcal{A}  \to \mathcal{B}$ and $\mathcal{B}$ is a discrete set.
\end{lemma}
\begin{proof}
	According to definition of conditional entropy, 
	\begin{align}
		H(A\mid f(A)) &= \sum_{b\in \{f(a)\colon a \in \mathcal{A}\}} H(A\mid f(A) = b) \Pr(f(A) = b) \\
		&= \sum_{b\in \{f(a)\colon a \in \mathcal{A}\}} H(A\mid f(A) = b) \sum_{a\colon f(a) = b} P_A(a) \\
		&  = \sum_{b\in \{f(a)\colon a \in \mathcal{A}\}} \sum_{a\colon f(a) = b}P_A(a)\times  H(A\mid f(A) = b)  \\
		&  = \sum_{b\in \{f(a)\colon a \in \mathcal{A}\}} \sum_{a\colon f(a) = b} P_A(a)  \times H(A\mid f(A) = f(a)) \\
		& = \sum_{a \in \mathcal{A}} P_A(a) \times H(A\mid f(A) = f(a))  
	\end{align}
\end{proof}
Next we have
\begin{align}
	&H\big(Y_{k-1}^{[n]} \mid  Y_k^{[n]}, {\bf G}_k^{[n]} \big) \notag \\
	&=\sum_{{\bf g}_k^{[n]} \in ({\mathbb{F}_q^\times}^n)^{k}}P_{{\bf G}_k^{[n]}}({\bf g}_k^{[n]})H\big(Y_{k-1}^{[n]} \mid  Y_k^{[n]}, {\bf G}_k^{[n]} ={\bf g}_k^{[n]}\big)\\
	&= \sum_{{\bf g}_k^{[n]} \in ({\mathbb{F}_q^\times}^n)^{k}}P_{{\bf G}_k^{[n]}}({\bf g}_k^{[n]})H\big(Y_{k-1}^{[n]} \mid  \chi(Y_{k-1}^{[n]}, {\bf g}_k^{[n]}), {\bf G}_k^{[n]} ={\bf g}_k^{[n]}\big) \\
	&=\sum_{{\bf g}_k^{[n]} \in ({\mathbb{F}_q^\times}^n)^{k}}P_{{\bf G}_k^{[n]}}({\bf g}_k^{[n]}) \times \notag \\
	&\hspace{1cm} \sum_{y_{k-1}^{[n]}\in \mathbb{F}_q^n}P_{Y_{k-1}^{[n]}}(y_{k-1}^{[n]}) H\big(Y_{k-1}^{[n]} \mid  \chi(Y_{k-1}^{[n]}, {\bf g}_k^{[n]}) = \chi(y_{k-1}^{[n]}, {\bf g}_k^{[n]}), {\bf G}_k^{[n]} ={\bf g}_k^{[n]}\big) ~~(\because \mbox{Lemma \ref{lem:cond_entropy_on_func}}) \label{eq:conv_use_cond_on_func} \\
	&\leq \sum_{y_{k-1}^{[n]}\in \mathbb{F}_q^n} \sum_{{\bf g}_k^{[n]} \in ({\mathbb{F}_q^\times}^n)^{k}} P_{Y_{k-1}^{[n]}}(y_{k-1}^{[n]})P_{{\bf G}_k^{[n]}}({\bf g}_k^{[n]}) \times \log_2 |\mathcal{S} (y_{k-1}^{[n]}, {\bf g}_k^{[n]})| \label{eq:conv_use_AIS} \\
	& \leq \mathbb{E}\big[ \log_2 |\mathcal{S} (Y_{k-1}^{[n]}, {\bf G}_k^{[n]})| \big]  \\
	& \leq \log_2 \mathbb{E} \big[ |\mathcal{S} (Y_{k-1}^{[n]}, {\bf G}_k^{[n]})| \big] ~~~~\mbox{\it (Jensen's inequality)} \\
	&\leq \log_2 \mathbb{E} \big[ |\mathcal{S} (\overline{y}_{k-1}^{[n]}, {\bf G}_k^{[n]})| \big], ~~~~ \exists \overline{y}_{k-1}^{[n]} \in \mathbb{F}_q^n~~~~(\because \eqref{eq:leqH}) \\
	&=  \log_2 \sum_{y_{k-1}^{[n]}\in \mathbb{F}_q^n} \Pr  \big\{ y_{k-1}^{[n]} \in \mathcal{S}(\overline{y}_{k-1}^{[n]}, {\bf G}_k^{[n]})    \big\} \label{eq:conv_exp_size_random_set} \\
	&= \log_2 \sum_{y_{k-1}^{[n]}\in \mathbb{F}_q^n} \Pr \big\{  {\chi}(y_{k-1}^{[n]}, {\bf G}_k^{[n]}) =  {\chi}(\overline{y}_{k-1}^{[n]}, {\bf G}_k^{[n]})  \big\} \label{eq:conv_probability_to_bound} \\ 
	&= \log_2 \sum_{y_{k-1}^{[n]}\in \mathbb{F}_q^n}\prod_{\tau=1}^n \Pr \big\{  {\chi}(y_{k-1}^{[n]}, {\bf G}_k^{[n]}) =  {\chi}(\overline{y}_{k-1}^{[n]}, {\bf G}_k^{[n]})  \big\} \label{eq:conv_G_time_indep} \\
	&= \log_2 \sum_{y_{k-1}^{[n]}\in \mathbb{F}_q^n} \prod_{\tau=1}^n \Pr \big\{G_{k1}^{(\tau)}(x_1^{(\tau)}-\overline{x}_1^{(\tau)}) +\cdots + G_{kk}^{(\tau)}(x_{k}^{(\tau)}-\overline{x}_{k}^{(\tau)}) =  0 \big\} \label{eq:conv_def_xs} \\
	&\leq \log_2 \sum_{y_{k-1}^{[n]}\in \mathbb{F}_q^n}\prod_{\tau=1}^n  \Big( \frac{1}{q-1}  \mathbb{I}\big( y_{k-1}^{(\tau)} \not= \overline{y}_{k-1}^{(\tau)} \big) + \mathbb{I}\big(y_{k-1}^{(\tau)} = \overline{y}_{k-1}^{(\tau)} \big) \Big) \label{eq:conv_xtoy} \\
	&= \log_2 \prod_{\tau=1}^n \sum_{y_{k-1}^{(\tau)}\in \mathbb{F}_q}  \Big( \frac{1}{q-1}  \mathbb{I}\big( y_{k-1}^{(\tau)} \not= \overline{y}_{k-1}^{(\tau)} \big) + \mathbb{I}\big(y_{k-1}^{(\tau)} = \overline{y}_{k-1}^{(\tau)} \big) \Big) \label{eq:conv_sum_product}  \\
	&\leq \log_2 \prod_{\tau=1}^n 2 \label{eq:conv_2qetaq}  \\
	&=n
\end{align}
Step \eqref{eq:conv_use_cond_on_func} is by applying Lemma \ref{lem:cond_entropy_on_func} by considering $\chi(Y_{k-1}^{[n]}, {\bf g}_k^{[n]})$ as the function $f(Y_{k-1}^{[n]})$. Note that the lemma is applied with the additional condition ${\bf G}_k^{[n]} = {\bf g}_k^{[n]}$. Also note that $Y_{k-1}^{[n]}$ is independent of ${\bf G}_k^{[n]}$ and therefore conditioning on ${\bf G}_k^{[n]} = {\bf g}_k^{[n]}$ does not change the distribution of $Y_{k-1}^{[n]}$, which is always $P_{Y_{k-1}^{[n]}}$.
To see Step \eqref{eq:conv_use_AIS}, first recall the definition of AIS in \eqref{eq:def_AIS}. Then the condition  $\chi(Y_{k-1}^{[n]}, {\bf g}_k^{[n]}) = \chi(y_{k-1}^{[n]}, {\bf g}_k^{[n]})$ implies $Y_{k-1}^{[n]}$ can only take values in the set $\mathcal{S}(y_{k-1}^{[n]}, {\bf g}_k^{[n]})$. The step then follows from the fact that the entropy of any discrete random variable $A\in \mathcal{A}$ must satisfy $H(A) \leq \log_2 |\mathcal{A}|$.
Step \eqref{eq:conv_exp_size_random_set} follows as the expectation of the cardinality of a random set is equal to the sum of the probabilities of each possible element being in the random set. 
Step \eqref{eq:conv_probability_to_bound} is due to the definition of AIS \eqref{eq:def_AIS}, since $\mathcal{S}(\bar{y}_{k-1}^{[n]}, {\bf G}_k^{[n]})$ is precisely the set of those values of $\gamma\in \mathbb{F}_q^n$ that yield $\chi(\gamma, {\bf G}_k^{[n]}) = \chi(\bar{y}_{k-1}^{[n]}, {\bf G}_k^{[n]})$.
Step \eqref{eq:conv_G_time_indep} is because ${\bf G}_k^{(\tau)}$ is independent across $\tau\in [n]$, and thus $Y_k^{(\tau)}$ {(recall that $Y_k^{[n]} \triangleq \chi(y_{k-1}^{[n]}, {\bf G}_k^{[n]})$) is independent across $\tau\in [n]$ for a fixed $y_{k-1}^{[n]}\in \mathbb{F}_q^n$. In Step \eqref{eq:conv_def_xs}, we define $x_l^{[n]} \triangleq \psi_l(y_{k-1}^{[n]})$ and $\overline{x}_l^{[n]} \triangleq \psi_l(\overline{y}_{k-1}^{[n]})$ for $l\in [k]$. In Step \eqref{eq:conv_xtoy}, $\mathbb{I}(X)$ is the indicator function, i.e., it returns $1$ if $X$ is true and $0$ otherwise. 
To see Step \eqref{eq:conv_xtoy}, first note that if $y_{k-1}^{(\tau)} \not= \overline{y}_{k-1}^{(\tau)}$, then $(x_1^{(\tau)},\cdots, x_k^{(\tau)}) \not= (\overline{x}_1^{(\tau)},\cdots, \overline{x}_k^{(\tau)})$. This is because  $(x_1^{(\tau)},\cdots, x_k^{(\tau)}) = (\overline{x}_1^{(\tau)},\cdots, \overline{x}_k^{(\tau)})$ implies $y_{k-1}^{(\tau)} = \overline{y}_{k-1}^{(\tau)}$.
Then note that if $(x_1^{(\tau)},\cdots, x_k^{(\tau)}) \not= (\overline{x}_1^{(\tau)},\cdots, \overline{x}_k^{(\tau)})$, the probability of the event in \eqref{eq:conv_def_xs} is upper bounded by $1/(q-1)$. This is argued as follows. If $\exists i\in [k],x_i^{(\tau)} - \overline{x}_i^{(\tau)} \not=0$, then conditioned on any realization of $(G_{kj}^{(\tau)})_{j\not= i}$, the event in \eqref{eq:conv_def_xs} is a linear equation on $G_{ki}^{(\tau)}$ that has at most one solution (because the coefficient for $G_{ki}^{(\tau)}$ is non-zero), and therefore the probability is upper bounded by $1/(q-1)$, because $G_{ki}^{(\tau)}$ is uniformly distributed over $\mathbb{F}_q^\times$. 
The above argument shows that $\Delta_k \leq n = n o_q(\log_2 q)$. 
This completes the proof of \eqref{eq:conv_to_prove_Delta_k}. \hfill \qed

\begin{remark}
	Note that the AIS bound argument holds if one assumes that the conditional p.m.f. of each  non-zero channel coefficient (given all other channel coefficients) is bounded by $\eta_q$ such that $\lim_{q\to \infty}\frac{\log_2(q\eta_q)}{\log_2 q} = 0$. To see it, in \eqref{eq:conv_xtoy} replace $\frac{1}{q-1}$ with $\eta_q$. Then \eqref{eq:conv_2qetaq} becomes $\log_2 \prod_{\tau=1}^n ((q-1)\eta_q +1)$ which is upper bounded by $\log_2 \prod_{\tau=1}^n (2q\eta_q)$ as $1\leq (q-1)\eta_q$. Then according to the assumption that $\lim_{q\to \infty}\frac{\log_2(q\eta_q)}{\log_2 q} = 0$, we again arrive at $\Delta_k = n o_q (\log_2 q)$.
\end{remark}

\subsection{Proof of classical achievability: TDMA}\label{sec:tdma}
First let us note that the DoF value for any user can be at most $d_k = 1$. This is true for both the $\mathbb{F}_q$ and the real Gaussian model, as is implied by the corresponding point-to-point communication results. In a nutshell, the proof follows from a `graph-burning' argument. Given any DoF tuple $(d_1,\cdots, d_K)$ that satisfies the region in \eqref{eq:tree_region_classical}, consider the progressive burning of the tree graph, starting from the root-node, such that $d_k$ is the amount of time it takes for the vertex Tx-$k$ to burn. Once a vertex is burnt, the fire spreads instantly to all the children of that vertex, whose burning times are determined by their assigned DoF values. The root node takes zero time to burn. The burning pattern yields a TDMA schedule, wherein Tx-$k$ is active only during the time that its corresponding vertex in the tree-graph is burning. By the nature of a tree graph, and how the fire progresses down the tree,  it is easy to see that at any time there can be at most one burning vertex in any path from a leaf node to a root node. This means that of all the transmit antennas that are connected to a receiver, at most one can be active at any time, corresponding to an orthogonal scheduling pattern (TDMA). The time it takes for each path to burn completely is exactly the sum of DoF values of the vertices along that path, which is bounded by $1$ (corresponding to $1$ DoF) according to \eqref{eq:tree_region_classical} for every path. 

Algorithm \ref{alg:TDMA} explicitly specifies for $k\in [K]$ the time interval in which Tx-$k$ is active and is used for serving Rx-$k$ only. The input for the algorithm is any tuple $(d_1,d_2,\cdots, d_K)$ that satisfies $d_{b_1}+d_{b_2}+\cdots+d_{b_L} \leq 1$ for every root-to-leaf path $(\mbox{Tx-}0, \mbox{Tx-}b_1, \cdots, \mbox{Tx-}b_L)$, along with the channel connectivity tree $\mathcal{T} = (\{\mbox{Tx-}0, \mbox{Tx-}1,\cdots, \mbox{Tx-}B\}, \mathcal{E})$. The output is the intervals $I_k \subseteq [0,1]$,  $k\in [K]$, identifying the time interval in which Tx-$k$ is active, and is used for serving only Rx-$k$. An example is illustrated in  Figure \ref{fig:TDMA}.  

\begin{algorithm} 
\caption{TDMA scheduling for the tree network}\label{alg:TDMA}
\renewcommand{\algorithmicrequire}{\textbf{Input:}}
\renewcommand{\algorithmicensure}{\textbf{Output:}}
\begin{algorithmic}
\Require $(d_1,d_2,\cdots, d_K)$, $\mathcal{T}$
\Ensure $(I_1,I_2,\cdots I_K)$
\For{$(\mbox{Tx-}0, \mbox{Tx-}b_1, \cdots, \mbox{Tx-}b_L)$ being a root-to-leaf path}
\State $t \gets 0$
\For{$k \gets b_1, b_2, \cdots, b_L$}
	\State $I_k \gets [t, t+d_k]$
	\State $t \gets t+d_k$
\EndFor
\EndFor
\end{algorithmic}
\end{algorithm}

\begin{figure}[h]
\center
\begin{tikzpicture}
[
    level 1/.style={sibling distance=25mm},
    level 2/.style={sibling distance=25mm},
    level 3/.style={sibling distance=15mm},
]
\node {$(\mbox{Root Tx-}0)$}
		child {
		node {$\mbox{Tx-}1$}
		child {
		node {$\mbox{Tx-}2$}
		child {node {$\mbox{Tx-}3$}} 
		child {node {$\mbox{Tx-}4$}} 
		} 
		child {node {$\mbox{Tx-}5$}
		child {node {$\mbox{Tx-}6$}}
		}
		}
		;
\node at (-3,-2) {$\mathcal{T}=$};
\begin{scope}[shift= {(5,-2.75)}]
	\node [align = left] {\small $I_1 = [0,d_1]$ \\
	\small $I_2 = [d_1, d_1+d_2]$\\
	\small $I_3 = [d_1+d_2, d_1+d_2+d_3]$\\
	\small $I_4 = [d_1+d_2, d_1+d_2+d_4]$\\
	\small $I_5 = [d_1, d_1+d_5]$ \\
	\small $I_6 = [d_1+d_5, d_1+d_5+d_6]$
	};
\end{scope}
\end{tikzpicture}
\caption{An example of a tree network and its TDMA scheduling based on Algorithm \ref{alg:TDMA}.} 
\label{fig:TDMA}
\end{figure}
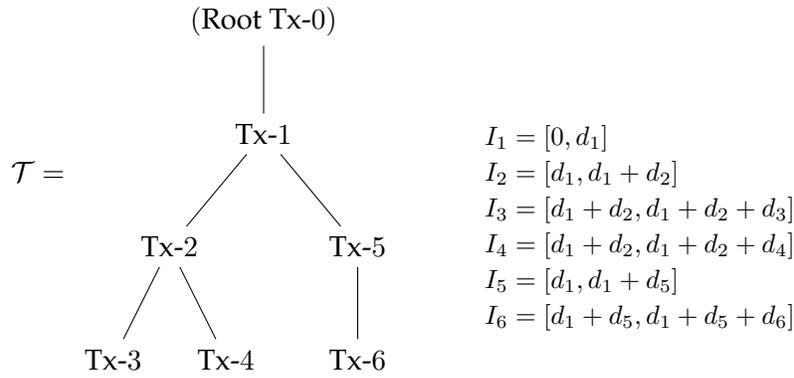

\bibliographystyle{IEEEtran}
\bibliography{../../bib_file/yy.bib}
\end{document}